\documentclass[reqno,12pt]{amsart}

\usepackage{color} 
\usepackage{ifpdf}
\ifpdf
    \usepackage[pdftex]{graphicx}
    \usepackage[pdftex]{hyperref}
    \hypersetup{
        unicode=false,          
        pdftoolbar=true,        
        pdfmenubar=true,        
        pdffitwindow=false,     
        pdfstartview={FitH},    
        pdftitle={MCP Article},      
        pdfauthor={Michael Holst},   
        pdfsubject={Mathematics},    
        pdfcreator={Michael Holst},  
        pdfproducer={Michael Holst}, 
        pdfkeywords={PDE, analysis, mathematical physics}, 
        pdfnewwindow=true,      
        colorlinks=true,        
        linkcolor=red,          
        citecolor=blue,         
        filecolor=magenta,      
        urlcolor=cyan           
    }

\else
    \usepackage{graphicx}
    \usepackage{pstricks}
    
    \newcommand{\href}[2]{#2}
\fi

\usepackage{times}
\usepackage{amsfonts}

\usepackage{amsmath}
\usepackage{amsthm}
\usepackage{amssymb}
\usepackage{amsbsy}
\usepackage{amscd}

\usepackage{enumerate}
\usepackage{verbatim}
\usepackage{subfigure}

\newtheorem{theorem}{Theorem}[section]

\newtheorem{proposition}[theorem]{Proposition}

\newtheorem{definition}[theorem]{Definition}

\newtheorem{remark}[theorem]{Remark}

\numberwithin{equation}{section}  

  \newcounter{mnote}
  \let\oldmarginpar\marginpar
    \renewcommand\marginpar[1]{\-\oldmarginpar[\raggedleft\footnotesize #1]%
    {\raggedright\footnotesize #1}}

\definecolor{myblue}{rgb}{0.2,0.2,0.7}
\definecolor{mygreen}{rgb}{0,0.6,0}
\definecolor{mycyan}{rgb}{0,0.6,0.6}
\definecolor{myred}{rgb}{0.9,0.2,0.2}
\definecolor{mymagenta}{rgb}{0.9,0.2,0.9}
\definecolor{mywhite}{rgb}{1.0,1.0,1.0}
\definecolor{myblack}{rgb}{0.0,0.0,0.0}

\newcommand{\beq}{\begin{equation}}
\newcommand{\eeq}{\end{equation}}
\newcommand{\beqa}{\begin{eqnarray}}
\newcommand{\eeqa}{\end{eqnarray}}
\newcommand{\cA}{{\mathcal A}}

\newcommand{\cH}{{\mathcal H}}

\newcommand{\cL}{{\mathcal L}}
\newcommand{\cM}{{\mathcal M}}

\newcommand{\cO}{{\mathcal O}}

\newcommand{\cS}{{\mathcal S}}
\newcommand{\cT}{{\mathcal T}}

\newcommand{\cY}{{\mathcal Y}}

\newcommand{\bW}{{\bf W}}

\newcommand{\bh}{{\bf h}}

\newcommand{\bj}{{\bf j}}

\newcommand{\bw}{{\bf w}}
\newcommand{\bx}{{\bf x}}

\def\ee{\epsilon}

\def\la{\lambda}

\def\La{\Lambda}

\begin{document}

\title[An Alternative Between Non-unique and Negative Yamabe Solutions]{An Alternative Between Non-unique and Negative Yamabe Solutions to the Conformal Formulation of the Einstein Constraint Equations}

\author[M. Holst]{Michael Holst}

\author[C. Meier]{Caleb Meier}

\address{Department of Mathematics\\
         University of California San Diego\\ 
         La Jolla CA 92093}

\thanks{MH was supported in part by 
        NSF Awards~1065972, 1217175, and 1262982.}
\thanks{CM was supported in part by NSF Award~1065972.}
\thanks{Email: 
         \hspace*{0.2cm} {\tt mholst@math.ucsd.edu},
         \hspace*{0.2cm} {\tt meiercaleb@gmail.com}
        }

\date{\today}


\begin{abstract}
The conformal method has been effective for parametrizing solutions to the 
Einstein constraint equations on closed $3$-manifolds.
However, it is still not well-understood; for example, existence of
solutions to the conformal equations for zero or negative Yamabe metrics 
is still unknown without the so-called ``CMC'' or ``near-CMC'' assumptions.
The first existence results without such assumptions, termed the 
``far-from-CMC'' case, were obtained by Holst, Nagy, and Tsogtgerel in 2008
for positive Yamabe metrics.
However, their results are based on topological arguments, and as a result 
solution uniqueness is not known.
Indeed, Maxwell gave evidence in 2011 that far-from-CMC solutions 
are not unique in certain cases.
In this article, we provide further insight by establishing a type of 
alternative theorem for general far-from-CMC solutions.
For a given manifold $\cM$ that admits a metric of positive 
scalar curvature and scalar flat metric $g_0$ with no conformal Killing 
fields, we first prove existence of an analytic, one-parameter family 
of metrics $g_{\la}$ through $g_0$ such that $R(g_{\la}) = \la$.
Using this family of metrics and given data $(\tau,\sigma, \rho,\bj)$,
we form a one-parameter family of operators $F((\phi,\bw),\la)$ 
whose zeros satisfy the conformal equations.
Applying Liapnuov-Schmidt reduction, we determine an analytic solution curve 
for $F((\phi,\bw),\la) = 0$ through a critical point where the linearization 
of $F((\phi,\bw),\la)$ vanishes.
The regularity of this curve, the definition of $F((\phi,\bw),\la)$, and 
the earlier far-from-CMC results of Holst et al.\ allow us to then prove
the following alternative theorem for far-from-CMC solutions:
either
(1) there exists a $\la_1 >0$ such that (positive Yamabe) solutions to the 
conformal equations are non-unique with data 
$(g_{\la_1},\la_1^2\tau, \la_1^2\sigma,\la_1^2\rho,\la_1^2\bj)$;
or 
(2) there exists $ \la_2 < 0$ such that (negative Yamabe) solutions to the 
conformal equations exist with data 
$(g_{\la_2},\la_2^2\tau, \la_2^2\sigma,\la_2^2\rho,\la_2^2\bj)$.
\end{abstract}

\maketitle


\vspace*{-1.2cm}
\tableofcontents

\section{Introduction}

The Einstein field equation $G_{\mu\nu} = \kappa T_{\mu\nu}$ can be formulated 
as a Cauchy problem where the initial data consists of a Riemannian metric 
$\hat{g}_{ab}$ and a symmetric tensor $\hat{k}_{ab}$ on a specified 
$3$-dimensional manifold $\cM$ \cite{HE75, RW84}.
However, one is not able to freely specify such initial data.
Like Maxwell's equations, the initial data $\hat{g}_{ab}$ and 
$\hat{k}_{ab}$ must satisfy constraint equations, where the constraints 
take the form
\begin{align}
\hat{R}+\hat{k}^{ab}\hat{k}_{ab}+\hat{k}^2 = 2\kappa\hat{\rho}, \label{eq1:5aug12}\\
\hat{D}_b\hat{k}^{ab}-\hat{D}^a\hat{k} = \kappa \hat{j}^a. \label{eq2:5aug12}
\end{align}
Here, $\hat{R}$ and $\hat{D}$ are respectively the scalar curvature and covariant derivative associated with $\hat{g}_{ab}$, $\hat{k}$ is the trace of
$\hat{k}_{ab}$, and $\hat{\rho}$ and $\hat{j}^a$ are matter terms obtained by contracting $T_{\mu\nu}$ with a vector field normal
to $\cM$, where one assumes that $T_{\mu\nu}$ satisfies the dominant energy condition.  

Equation \eqref{eq1:5aug12} is known as the Hamiltonian constraint while \eqref{eq2:5aug12} is known as the momentum
constraint, and collectively they are known as the Einstein constraint equations.
These equations form an underdetermined system of four equations to be solved for twelve unknowns represented by the symmetric two index tensors $\hat{g}_{ab}$ and $\hat{k}_{ab}$.  
In order to transform the constraint equations into a determined system, one divides the unknowns into freely specifiable data and determined data using what is known as the \emph{conformal method}.  In this method, introduced by
Lichnerowicz \cite{AL44} and York \cite{JY71}, one makes the decomposition 
\begin{align}
\hat{k}_{ab} = \hat{l}_{ab}+\frac13\hat{g}_{ab}\hat{\tau},
\end{align}
where $\hat{\tau}  = \hat{k}_{ab}\hat{g}^{ab}$ is the trace and $\hat{l}_{ab}$ is the traceless part of $\hat{k}_{ab}$, and then
one makes the following conformal rescaling
\begin{align}
\hat{g}_{ab} = \phi^4g_{ab}, \quad \hat{l}_{ab} = \phi^{-10}l^{ab}, \quad \hat{\tau} = \tau. 
\end{align}  
Then, forming the decomposition
\begin{align}
l_{ab} = (\sigma_{ab}+(\cL \bw)_{ab}),
\end{align}
where $D_a\sigma^{ab} = 0$, and defining
$$
(\cL \bw)^{ab} = D^aw^b + D^bw^a-\frac23(D_cw^c)g^{ab}
$$
as the \emph{conformal Killing operator},
one obtains the conformal, transverse, traceless (CTT) formulation of the constraint equations as
\begin{align}\label{eq1:12mar13}
-\Delta\phi + &\frac{1}{8}R \phi + \frac{\la^4}{12}\tau^2\phi^5-\frac{1}{8}(\sigma+\mathcal{L}{\bw})_{ab}(\sigma+\mathcal{L}\bw)^{ab}\phi^{-7}-\frac{\kappa}{4}\rho\phi^{-3}=0,\\
&\mathbb{L}\bw +\frac{2}{3}D\tau\phi^6+\la^2\kappa\bj =0, \nonumber
\end{align}
where $\mathbb{L} \bw= -D_b(\cL \bw)^{ab}$.
The above system \eqref{eq1:12mar13}
forms a determined, coupled nonlinear system of elliptic partial differential equations with specified data $(g, \tau, \sigma, \rho,\bj)$ and with $(\phi,\bw)$ to be determined by the equations. 
For simplicity, we will refer to this system as the \emph{conformal formulation} (cf. \cite{BI04} for further discussion). 

In this paper, we address some of the open questions associated with existence and uniqueness of solutions to the conformal formulation on a closed, $3$-dimensional manifold $\cM$ in the event that the mean curvature $\tau$ does not satisfy the ``near constant" (or \emph{near-CMC})
assumptions developed by Isenberg and Moncrief in \cite{IM96}.  It is well-known that solutions to the conformal equations exist and are unique 
on a closed manifold if
the mean curvature $\tau$ does not vanish and has a bounded derivative.  However, very little is known about the existence and uniqueness
of solutions in the event that the mean curvature function does not satisfy these so-called near-CMC assumptions.  
The first ``far-from-CMC" existence results were not established until 2008 in \cite{HNT08,HNT09}, when Holst, Nagy, and Tsogtgerel showed that solutions to the conformal formulation exist for metrics in the positive Yamabe class and mean curvatures $\tau$ completely free of the near-CMC assumption, now termed the ``far-from-CMC" case.
However, there are currently no far-from-CMC existence results for metrics in the zero or negative Yamabe classes.  Furthermore,
given that the existence results in \cite{HNT08,HNT09} use a general topological fixed point theorem as opposed to the contraction
mapping theorem type arguments used in \cite{JI95, IM96}, it is not known whether far-from-CMC solutions are unique.  Indeed, Maxwell has shown
that solutions to the conformal formulation are non-unique for certain low-regularity, far-from-CMC mean curvatures in the event that the prescribed metric
lies in the zero Yamabe class (cf \cite{MaD11}).  In this article we partially address these issues by showing that either the postive Yamabe, 
far-from-CMC solutions obtained in \cite{HNT08,HNT09} are non-unique, or that negative Yamabe, far-from-CMC solutions to the conformal equations
exist for a certain family of metrics with constant, negative scalar curvature.  

To obtain our results, we consider a closed, $3$-dimensional manifold $\cM$ which admits a metric of positive scalar curvature and also admits   
a metric $g_0$ with zero scalar curvature and no conformal Killing fields.
We show that there exists a $\delta>0$ and a one-parameter family of metrics $(g_{\la})_{\la \in (-\delta,\delta)}$ on $\cM$, analytic in the variable
$\la$, such that $R(g_{\la}) = \la$ and $g_{\la}|_{\la=0}= g_0$.  Using this family of metrics, we then construct the following one-parameter family of nonlinear elliptic 
systems on the closed manifold $\cM$: 
\begin{align}\label{eq1:26oct12}
-\Delta_{\la}\phi + &\frac{1}{8}\la \phi + \frac{\la^4}{12}\tau^2\phi^5-\frac{1}{8}(\la^2\sigma+\mathcal{L}{\bw})_{ab}(\la^2\sigma+\mathcal{L}\bw)^{ab}\phi^{-7}-\frac{\la^2\kappa}{4}\rho\phi^{-3}=0,\\
&\mathbb{L}_{\la}\bw +\frac{2\la^2}{3}D_{\la}\tau\phi^6+\la^2\kappa\bj =0, \nonumber
\end{align}
where $\Delta_{\la}$, $\mathbb{L}_{\la}$ and $D_{\la}$ are the Laplace-Beltrami operator, negative divergence of the conformal Killing operator and covariant derivative with respect to the metric $g_{\la}$.  
For a fixed $\la$, we recognize the above family as the CTT formulation
of the Einstein Constraint Equations with specified data
\begin{align}
g_{\la}, \quad \tau_{\la} = \la^2 \tau, \quad \sigma_{\la} = \la^2\sigma \quad \rho_{\la} = \la^2\rho, \quad \text{and} \quad {\bj}_{\la} = \la^2 \bj.
\end{align}   
We assume that $\tau$ is an arbitrary differentiable function on $\cM$, so that $\tau$ does not satisfy the near-CMC assumptions.  
By applying some basic techniques from bifurcation theory and nonlinear functional analysis to \eqref{eq1:26oct12}, 
we are able to parametrize the solution curve of \eqref{eq1:26oct12} through $((1,{\bf 0}),0)$.  An analysis of this 
solution curve reveals that, under suitable reasonable assumptions, at least one of the following two possibilities must occur:
\begin{enumerate}\label{eq1:6nov12}
\item There exists a $\delta >0$ such that for $\la_0 \in (0,\delta)$, there exist $(\phi_{1,\la_0},\bw_{1,\la_0})$ and $(\phi_{2,\la_0},\bw_{2,\la_0})$ in $C^{2,\alpha}\oplus C^{2,\alpha}(T\cM)$ 
that together solve \eqref{eq1:26oct12} when $\la = \la_0$ with $(\phi_{1,\la_0},\bw_{1,\la_0}) \ne (\phi_{2,\la_0},\bw_{2,\la_0})$ (i.e.\ solutions to the CTT formulation
are non-unique).
\item There exists a $\delta > 0$ such that for any $\la_0 \in (-\delta, 0)$, there exists $(\phi_{\la_0},\bw_{\la_0}) \in C^{2,\alpha}\oplus C^{2,\alpha}(T\cM)$ that 
solves \eqref{eq1:26oct12} when $\la = \la_0$ (i.e.\ far-from CMC solutions
to the CTT formulation exist for certain metrics in the negative Yamabe class).
\end{enumerate}

The remainder of the paper is organized as follows.
Section~\ref{prelim} presents notation and preliminaries
that we will require to prove our results.
In particular, we first summarize some fundamental results from bifurcation 
theory.
In particular, we discuss what is known as {\bf Liapunov-Schmidt} reduction,
which is instrumental in parametrizing solutions to \eqref{eq1:26oct12} in a 
neighborhood of $((1,{\bf 0}),0)$.
We then show that a closed, $3$-dimensional manifold $\cM$ which admits a 
metric of positive scalar curvature also admits an analytic, one-parameter 
family of metrics $g_{\la}$ such that $R(g_{\la}) = \la$.
In Section \ref{sec1:1nov12}, we then use this one-parameter family of metrics 
and given data $(\tau,\sigma, \rho, \bj)$ 
for the conformal equations to define a nonlinear operator $F((\phi,\bw),\la)$ whose zeroes coincide with solutions to
the conformal equations.  The main results of this paper are then presented in Theorems~\ref{thm1:31oct12} and \ref{thm2:31oct12} in Section~\ref{sec1:1nov12}.
Theorem~\ref{thm1:31oct12} characterizes the behavior of solutions to the nonlinear problem $F((\phi,\bw),\la) = 0$ in a neighborhood of the point $((1,{\bf 0}),0)$.
This characterization allows us to conclude that either there exists $\la_0> 0$ such that solutions to $F((\phi,\bw),\la_0) = 0$ are non-unique or that there exists $\la_0 <0$
for which solutions to $F((\phi,\bw),\la_0) = 0$ exist.  Theorem~\ref{thm2:31oct12} then interprets this result in terms of the conformal equations.
It concludes that in any neighborhood of a metric $g_0$ with zero scalar curvature and no conformal Killing fields on $\cM$, that either there exists 
a metric $g_{\la}$ with $R(g_{\la}) = \la >0$ for which solutions to the conformal
equations are non-unique, or $R(g_{\la}) = \la < 0$ and negative Yamabe, far-from-CMC solutions exist.
The remainder of the paper is then devoted to proving these results.  Section~\ref{analytic} is dedicated to showing
that the operator $F((\phi,\bw),\la)$ is analytic, and then in Section~\ref{mainproof} we prove Theorems~\ref{thm1:31oct12} and \ref{thm2:31oct12}.    
We draw some conclusions in Section~\ref{conc},
and also include Appendx~\ref{app} containing some supporting results.

\section{Preliminary Material}
\label{prelim}

\subsection{Notation and Function Spaces}\label{notation}

Let $\cM$ denote a compact $3$-dimensional manifold and let $T^{r}_s\cM$ denote the vector bundle of tensors of type $(r,s)$.
In this paper,
we will consider the space of $k$-differentiable sections $C^k(T^r_s\cM)$, 
the H\"{o}lder spaces $C^{k,\alpha}(T^r_s\cM)$ where $k \in \mathbb{N},~ p \ge 1, ~\alpha \in (0,1)$, and
the Sobolev spaces $W^{k,p}(T^r_s\cM)$. 
Note that all of these spaces (see Appendix~\ref{app} for a quick summary of the standard notation we use here for norms) are Banach spaces,
and the space $W^{k,2}(T^r_s)$ is a Hilbert space for all $k$.  As in \cite{FiMa75}, we let
\begin{align}
&\cS_2^{s,p} = W^{s,p}(T_{2,\text{symmetric}}^0(\cM)) \quad \text{the symmetric 2-covariant $W^{s,p}$ tensors},\nonumber\\
&\cA^{s,p} \subset \cS_2^{s,p} \quad \text{the open set of Riemannian metrics of type $W^{s,p}$ with $s > \frac{3}{p}$}.\nonumber
\end{align}  
We will denote scalar valued functions by simply writing $C^k$, $C^{k,\alpha}$ and $W^{s,p}$. 

Using any of the above Banach spaces, one can form new Banach spaces and Hilbert spaces by considering the direct
sum (see also~\cite{CMH12}).
\begin{definition}\label{def2:24july12}
Suppose that $X_1$ and $X_2$ are Banach spaces with norms $\|\cdot\|_{X_1}$ and $\|\cdot\|_{X_2}$.
Then
the direct sum $X_1\oplus X_2$ is the vector space of ordered pairs $(x,y)$ where $x\in X_1$, $y\in X_2$ and addition and scalar multiplication
are carried out component-wise.
\end{definition}
\noindent
We have the following proposition: 
\begin{proposition}\label{prop1:20aug12}
The vector space $X_1\oplus X_2$ is a Banach space when given the norm 
\begin{align}
\|(x,y)\|_{X_1\oplus X_2} =\left(\|x\|^2_{X_1}+\|y\|^2_{X_2}\right)^{\frac12}.
\end{align}
\end{proposition}
\begin{proof}
This follows from the fact that $\|\cdot\|_{X_1}$ and $\|\cdot\|_{X_2}$ are norms and
the spaces $X_1$ and $X_2$ are complete with respect to these norms.
\end{proof}
\noindent
We have a similar proposition for Hilbert spaces.
\begin{proposition}\label{def1:24july12}
Suppose that $\cH_1$ and $\cH_2$ are Hilbert spaces with inner products $\langle \cdot,\cdot \rangle_{\cH_1} $ and $\langle \cdot,\cdot \rangle_{\cH_2} $.
Then the direct sum $H_1\oplus H_2$ is a Hilbert space with inner product
\begin{align}
\langle (w,x),(y,z) \rangle_{\cH_1\oplus \cH_2}  =\langle w,y \rangle_{\cH_1} + \langle x,z \rangle_{\cH_2}.
\end{align}
\end{proposition}
\begin{proof}
That $\langle \cdot,\cdot \rangle_{\cH_1\oplus \cH_2}$ is an inner product follows from the fact
that $\langle \cdot,\cdot\rangle_{\cH_1} $ and $\langle \cdot,\cdot \rangle_{\cH_2}$ are inner products.
The expression 
$$\|(u,v),(u,v)\|_{\cH_1\oplus \cH_2} = \sqrt{\langle (u,v),(u,v) \rangle_{\cH_1\oplus \cH_2} },$$
is a norm on $\cH_1\oplus \cH_2$ that coincides with the norm in Proposition~\ref{prop1:20aug12}
in the event that the norms on $X_1$ and $X_2$ are induced by inner products.
\end{proof}
\noindent
See \cite{EZ86} for a more complete discussion about the direct sums of Banach spaces.

\subsection{Analytic Operators and the Implicit Function Theorem}

Here we briefly discuss analytic operators and the Implicit Function Theorem.
Our approach to proving that either negative Yamabe far-from-CMC solutions exist or
that positive Yamabe far-from-CMC solutions are non-unique relies on showing that the operator
in \eqref{eq1:12mar13} is analytic.  We then apply the Implicit Function Theorem to 
determine an analytic solution curve through a critical point where the linearization of \eqref{eq1:12mar13} has a nontrivial kernel.  To this end, the following discussion will
be essential going forward; the treatment is taken mostly from \cite{EZ86}.

Let $X$ and $Y$ be Banach spaces and assume that $M:X \times \cdots \times X \to Y$
is a $k$-linear bounded operator which is symmetric in all variables.  
We define a norm on $M$ by
\begin{align}\label{eq1:20may13}
\|M\| = \sup_{\|x_1\| = \cdots = \|x_n\| = 1} \|M(x_1,\cdots,x_n)\|,
\end{align}
which implies that 
$$
 \|M(x_1,\cdots,x_n)\| \le \|M\|\|x_1\|\|x_2\|\cdots \|x_n\| \quad \text{for all}~~ (x_1,\cdots, x_n).
$$

\begin{definition}\label{def2:25apr13}
A power operator can be created from $M$ by defining
\begin{align}
Mx^k &= M(x,\cdots,x),\\
Mx^my^n &= M(\underbrace{x,\cdots,x},\underbrace{y,\cdots,y}), \quad m+n = k  \nonumber,\\
& \quad \quad \quad m~\text{times} \quad n~\text{times} \nonumber
\end{align}
for any partition of $k$.  For $k=0$, $Mx^0$ will denote a fixed element in $X$.  
\end{definition}

Using this definition of power operator, we can then form operators of the form
\begin{align}\label{eq1:25apr13}
Tx = \sum_{n=0}^{\infty} T_n(x-x_0)^n,
\end{align}
where each $T_n$ is a power operator.  The operator $T$ converges
absolutely if the series
\begin{align}\label{eq2:25apr13}
\sum_{n=0}^{\infty} \|T_n\|\|x-x_0\|^n,
\end{align}
converges.

\begin{definition}\label{def1:25apr13}
Let $X$ and $Y$ be Banach spaces and let $T_n:X\to Y$ be power operators, $n \in \mathbb{N}$.  
\begin{itemize}
\item[(a)]  The operator $T: U \subset X \to Y$ is analytic at a point $x_0 \in X$ if and only if it is
defined on some neighborhood of $x_0$ and there is some number $r>0$ such that the series
\eqref{eq2:25apr13} converges for all $x$ with $\|x-x_0\| < r$.

\item[(b)] $T$ is analytic on the open set $U$ if and only if $T$ is analytic at every point of $U$.
\end{itemize}  
\end{definition}

A central theorem which we state without proof, and also taken in this particular form from \cite{EZ86}, is the Implicit Function Theorem.  

\begin{theorem}[Implicit Function Theorem]\label{thm1:25apr13}
Suppose that $X, Y$ and $Z$ are Banach spaces with $U \subset X\times Y$ a neighborhood of $(x_0,y_0)$.
Let $F:U\subset X\times Y \to Z$ be an operator satisfying $F(x_0,y_0) = 0$.  Then if
\begin{itemize}
\item[(i)] $D_y F$ exists on $U$ and $\text{ker}(D_yF(x_0,y_0))$ is trivial,

\item[(ii)] $F$ and $D_yF$ are continuous at $(x_0,y_0)$,

\end{itemize}
the following are true:

\begin{itemize}
\item[(a)]  There exist positive numbers $r_0$ and $r$ such that for every $x \in X$ satisfying
$\|x-x_0\| < r_0$, there is exactly one $y(x) \in Y$ for which $\|y(x) - y_0\| \le r$ and $F(x,y(x)) = 0$.

\item[(b)] If $F$ is a $C^m$-map, $1 \le m \le \infty$, on a neighborhood of $(x_0,y_0)$, then $y(x)$ is
also a $C^m$-map on a neighborhood $x_0$.

\item[(c)] If $F$ is analytic at $(x_0,y_0)$, then $y(x)$ is analytic at $x_0$.
\end{itemize}
\end{theorem}

\subsection{Basic Bifurcation Theory}\label{bifur}
We now present some basic concepts from bifurcation theory that will be also essential in our analysis.
The following treatment is taken from \cite{HK04} and \cite{CH82};
see also~\cite{StHo2011a}.

Suppose that $F: U\times V \to Z$ is a mapping with open sets $U \subset X,V \subset \Lambda$,
where $X$ and $Z$ are Banach spaces and $\Lambda = \mathbb{R}$.  We let $x\in X$ and $\lambda \in \La$.
Additionally assume that $F(x,\lambda)$ is Fr\'{e}chet differentiable with respect to $x$ and $\lambda$ on $U\times V$.
We are interested in solutions to the nonlinear problem
\begin{align}\label{eq3:2july12}
F(x,\la) = 0.
\end{align}
A solution of \eqref{eq3:2july12} is a point $(x,\la) \in X\times \La$ such
that \eqref{eq3:2july12} is satisfied. 

\begin{definition}
Suppose that $(x_0,\la_0)$ is a solution to \eqref{eq3:2july12}.  We say that $\la_0$ is a {\bf bifurcation point} if 
for any neighborhood $U$ of $(x_0,\la_0)$ there exists a $\la\in \La$ and $x_1, x_2 \in X$, $x_1 \ne x_2$
such that $(x_1,\la), (x_2,\la) \in U$ and $(x_1,\la)$ and $(x_2,\la)$ are both solutions to \eqref{eq3:2july12}.
\end{definition}

Given a solution $(x_0,\la_0)$ to \eqref{eq3:2july12}, we are interested in analyzing solutions to \eqref{eq3:2july12}
in a neighborhood of $(x_0,\la_0)$ to determine whether or not it is a bifurcation point.  One of the most useful tools for this is the Implicit Function Theorem~\ref{thm1:25apr13}.  
This theorem asserts that if $D_xF(x_0,\la_0)$ is invertible, then there exists a 
neighborhood $U_1\times V_1 \subset U\times V$ and a continuous function $f: V_1\to U_1$ such that all solutions
to \eqref{eq3:2july12} in $U_1\times V_1$ are of the form $(f(\la),\la)$.  Therefore, in order for a bifurcation to occur at $(x_0,\la)$,
it follows that $D_xF(x_0,\la_0)$ must not be invertible.

\subsubsection{Liapunov-Schmidt Reduction}\label{LSchmidt}

The following discussion is taken from \cite{HK04}.  Let $X, \Lambda$ and $Z$ be Banach spaces and assume that $U\subset X$, $V\subset \La$.  
For $\la = \la_0$, we require that the mapping 
${F:U\times V \to Z}$ be a 
nonlinear Fredholm operator with respect to $x$; i.e.\ the linearization $D_xF(\cdot,\la_0)$ of $F(\cdot,\la_0):U \to Z$ 
is a Fredholm operator.  Assume that $F$ also satisfies the following assumptions:
\begin{align}\label{eq7:2july12}
&F(x_0,\la_0) = 0 \quad \text{for some $(x_0,\la_0) \in U \times V$},\\
&\text{dim ker}(D_xF(x_0,\la_0)) = \text{dim ker}(D_xF(x_0,\la_0)^*) = 1. \nonumber
\end{align}
Given that $D_xF(x_0,\lambda_0)$ has a one-dimensional kernel, 
there exists a projection operator $P:X \to X_1 = \text{ker}(D_xF(x_0,\la_0))$.  Similarly, one has the projection operator
${Q:Y\to Y_2 = \text{ker}(D_xF(x_0,\la_0)^*)}$.  This allows us to decompose $X = X_1 \oplus X_2$ and $Y = Y_1 \oplus Y_2$ where 
$Y_1 = R(D_XF(x_0,\lambda_0))$.  We will refer to the decomposition $X_1\oplus X_2$ and $Y_1\oplus Y_2$ induced by
$D_xF(x_0,\la_0)$ as the {\bf Liapunov decomposition}, and we see that  
$F(x,\lambda) = 0$ if and only if the following two equations are satisfied
\begin{align}\label{eq1:27june12}
&QF(x,\lambda) = 0,\\
&(I-Q)F(x,\lambda) = 0. \nonumber
\end{align}

For any $x\in X$, we can write $x = v+w$, where $v= Px$ and $w = (I-P)x$. 
Define $G: U_1\times W_1 \times V_1 \to Y_1$ by
\begin{align}\label{eq5:2july12}
&G(v,w,\la) = (I-Q)F(v+w,\la), \quad \text{where}\\
&U_1\subset X_1,\hspace{3mm} W_1 \subset X_2, \hspace{3mm}   V_1 \subset \mathbb{R} \quad \text{and} \nonumber \\
&v_0 = Px_0 \in U_1, \quad w_0 = (I-P)x_0 \in W_1, \nonumber
\end{align}
and $U_1, W_1$ are neighborhoods such that $U_1 + W_1 \subset U \subset X$.

Then the definition of $G(v,w,\la)$ implies that $G(v_0,w_0,\la_0) = 0$ and our choice of function spaces ensures that
$$D_wG(v_0,w_0,\la_0) =(I- Q)D_xF(x_0,\la_0):X_2 \to Y_1,$$
is bijective.  The Implicit Function Theorem~\ref{thm1:25apr13} then implies that there exist
neighborhoods $U_2\subset U_1, W_2\subset W_1$ and $V_2\subset V_1$ and
a continuous function  
\begin{align}\label{eq9:6july12}
&\psi:U_2\times V_2 \to W_2 \quad \text{such that all solutions to $G(v,w,\la) = 0$}, \\
&\text{ in $U_2\times W_2\times V_2$ $\quad$ are of the form $G(v,\psi(v,\la), \la ) = 0.$}\nonumber 
\end{align}
Insertion of $\psi(v,\la)$ into the second equation in \eqref{eq1:27june12} yields
a finite-dimensional problem
\begin{align}\label{eq6:2july12}
\Phi(v,\la) = QF(v+\psi(v,\la),\la) = 0.
\end{align}
We observe that finding solutions $(v,\la)$ to \eqref{eq6:2july12} is equivalent to finding
solutions to $F(x,\la) = 0$ in a neighborhood of $(x_0,\la_0)$.  We will refer to the finite-dimensional problem \eqref{eq6:2july12}
as the {\bf Liapunov-Schmidt reduction} of \eqref{eq3:2july12}. 

With additional assumptions on the operator $F(x,\la)$ and another application of the Implicit Function Theorem, we may 
conclude that all solutions to \eqref{eq6:2july12} are of the form
\begin{align}\label{eq2:26apr13}
(v,\gamma(v)), \quad \gamma: U_3\subset U_2 \to I \subset \mathbb{R}.
\end{align}
Therefore, all solutions to \eqref{eq6:2july12}
in a neighborhood of $v_0$ must satisfy 
\begin{align}\label{eq1:26apr13}
g(v) = QF(v+\psi(v,\gamma(v)),\gamma(v)) = 0.
\end{align} 
Given that $\text{ker}(D_xF(x_0,\lambda_0))$ is spanned by $\hat{v}_0$, then we can write $v = s\hat{v}_0 +v_0$.
Substituting this into \eqref{eq1:26apr13} we obtain
\begin{align}\label{eq2:27jun12}
g(s) = QF(s\hat{v}_0+v_0+\psi(s\hat{v}_0+v_0,\gamma(v_0+s\hat{v}_0),\gamma(v_0+s\hat{v}_0) = 0.
\end{align}  
This reduction provides the basis of the following theorem taken from \cite{HK04}, which allows us to determine
a unique solution curve through the point $(x_0,\la_0)$. 

\begin{theorem}\label{thm1:29apr12}
Assume $F:U\times V \to Z$ is continuously differentiable on ${U\times V \subset X \times \mathbb{R}}$ and
that assumptions \eqref{eq7:2july12} hold. 
Additionally, assume that 
\begin{align}\label{eq3:29apr12}
D_{\lambda}F(x_0,\lambda_0) \notin R(D_xF(x_0,\lambda_0)).
\end{align}
Then there is a continuously differentiable curve through $(x_0,\lambda_0)$.
That is, there exists
\begin{align}\label{eq4:29apr12}
\{(x(s),\lambda(s))~|~s \in (-\delta, \delta),~ (x(0),\lambda(0)) = (x_0,\lambda_0)\},
\end{align}
such that
\begin{align}\label{eq5:29apr12}
F(x(s), \lambda(s)) = 0 \quad \text{for $s\in (-\delta, \delta)$},
\end{align}
and all solutions of $F(x,\lambda) = 0$ in a neighborhood of $(x_0, \lambda_0)$ belong to the curve
\eqref{eq4:29apr12}.
\end{theorem}
\begin{proof}
See \cite{CMH12} or \cite{HK04}.
\end{proof}

In order to demonstrate that a nonlinear operator $F(x,\la)$ exhibits a bifurcation point and has non-unique
solutions to $F(x,\la) = 0$, one constructs the solution curve in Theorem~\ref{thm1:29apr12} through a
point $(x_0,\la_0)$ where $D_xF(x_0,\la_0)$ has a nontrivial, one-dimensional kernel.  One then analyzes
the coefficients in the Taylor expansion of this solution curve at the critical points $(x_0,\la_0)$ 
using additional results from bifurcation theory to determine if it has a ``fold".  
We will not employ this approach in our paper, as the operator $F((\phi,\bw),\la)$ in \eqref{eq1:26oct12}
is not amenable such techniques.  
(However, see our related work in~\cite{CMH12}.)

Instead, we rely on additional regularity of our
solution curve in \eqref{eq4:29apr12}.  In particular, we demonstrate that
our solution curve is analytic in a neighborhood of $0$.  The far-from-CMC existence results \eqref{farCMC} combined with
the analyticity of our curve will allow us to conclude that $\la(s)$ cannot vanish identically in a neighborhood of
zero.  This is the crux of our argument.  To demonstrate the analyticity of our solution curve, we must show that
the one-parameter family $g_{\la}$ defined above \eqref{eq1:26oct12} is analytic in $\la$ in a neighborhood of zero.
This will allow us to conclude that the operator $F((\phi,\bw),\la)$ in \eqref{eq1:26oct12} is analytic in a neighborhood of the critical point $((1,{\bf 0}),0)$,
and therefore that our solution curve is analytic by the Implicit Function Theorem.  
We first prove the existence of the analytic, one-parameter family $g_{\la}$ for closed, $3$-dimensional manifolds $\cM$ that admit a metric
with positive scalar curvature.

\subsection{Properties of the Scalar Curvature Operator}\label{curvature}

The scalar curvature operator
$$
R: \cA^{s,p} \to W^{s-2,p},
$$
takes the form
\begin{align}\label{eq1:11mar13}
\left.R(g)\right|_{U_{i_j}}= &-\frac12g^{ij}g^{ab}\frac{\partial^2g_{ij}}{\partial x^a\partial x^b}+\frac12g^{ij}g^{ab}\frac{\partial^2g_{ai}}{\partial x^b\partial x^j} +\frac12g^{ij}g^{ab}\frac{\partial^2g_{aj}}{\partial x^b\partial x^i}\\
 &-\frac12g^{ij}g^{ab}\frac{\partial^2g_{ab}}{\partial x^i\partial x^j} -g^{ij}g_{ab}g^{kl}\Gamma_{ij}^{~~a}\Gamma_{kl}^{~~b}+g^{ij}g^{ab}g_{kl}\Gamma_{ai}^{~~k}\Gamma_{bj}^{~~l}, \nonumber
\end{align}
where $U_j$ is a given coordinate chart and $g \in \cA^{s,p}$.  The main objective of this section is to show that for a given manifold $\cM$ which admits a metric of positive scalar curvature, 
that there exists an analytic one-parameter family of metrics $(g_{\la})$ on $\cM$ that satisfies $R(g_{\la}) = \la$ for $\la \in (-\delta, \delta)$.  This 
family of metrics is necessary for the construction of the one-parameter family of non-linear problems in \eqref{eq1:26oct12}.
 
Using the definition of $R(g)$, we have the first preliminary result.

\begin{theorem}\label{thm1:26oct12}
The scalar curvature operator $R:\cA^{s,p} \to W^{s-2,p}$ is an analytic operator.
\end{theorem}
\begin{proof}
We first note that the scalar curvature operator is a smooth operator \cite{FiMa75}.
Fix a metric $g_0 \in \cA^{s,p}$.  Then for any $w \in \cA^{s,p}$, let $h = w- g_0$.  Then by Theorem~\ref{thm1:29apr13}, the remainder term $R^n$ for
the $n$-th order Taylor series about $g_0$ has the form
\begin{align}
\|R^{n}(w)\|_{W^{s,p}} \le \frac{1}{(n)!}\sup_{0<\tau<1}\|D^{(n)}R(g_0+\tau h)(h)^{n}\|_{W^{s,p}(\cM)},
\end{align}
where $D^n$ is the $n$-th Frechet derivative of $R$ and $h^n = (h, \cdots,h)$ is an 
element of $(\cA^{s,p})^n$.  See \cite{EZ86} for more details.  If $(\rho_i,U_i)$ is a coordinate chart of $\cM$, let $(\chi_j)_{j=1}^N$ denote
a smooth partition of unity subordinate to the $U_i$.  Then we have that
\begin{align}
\|D^{(n)}R(g_0+ & \tau h)(h)^{n}\|_{W^{s,p}(\cM)}
\label{eqn:dnr}
\\
& 
  \le \sum_{j=1}^N \|\chi_jD^{(n)}R(g_0+\tau h)(h)^{n}\|_{W^{s,p}(U_{i_j})},
\end{align}
where $\text{supp}(\chi_j) \subset U_{i_j}$.  In each chart $U_{i_j}$, we have that
\begin{align}
\left.R(g)\right|_{U_{i_j}}= &-\frac12g^{ij}g^{ab}\frac{\partial^2g_{ij}}{\partial x^a\partial x^b}+\frac12g^{ij}g^{ab}\frac{\partial^2g_{ai}}{\partial x^b\partial x^j} +\frac12g^{ij}g^{ab}\frac{\partial^2g_{aj}}{\partial x^b\partial x^i}\\
 &-\frac12g^{ij}g^{ab}\frac{\partial^2g_{ab}}{\partial x^i\partial x^j} -g^{ij}g_{ab}g^{kl}\Gamma_{ij}^{~~a}\Gamma_{kl}^{~~b}+g^{ij}g^{ab}g_{kl}\Gamma_{ai}^{~~k}\Gamma_{bj}^{~~l}.
\end{align}
In local coordinates, 
$$
\left.R(g_{ij})\right|_{U_{i_j}} :W^{s,p}(T^0_2(U_{i_j})) \to W^{s-2,p}(U_{i_j}),
$$
and 
$$
D^{k}(\left.R(g_{ij})\right|_{U_{i_j}}) \equiv 0
$$
for $k\ge 8$.  
This together with \eqref{eqn:dnr} implies the result.
\end{proof}

We will also have need for the following theorem from \cite{FiMa75}, which allows us to decompose
$S^{s,p}_2$ using the linearization of $R$ at a non-flat metric $g_0 \in \cA^{s,p}$.  Recall that on a
$3$-dimensional manifold $\cM$, non-flat (non-vanishing curvature tensor) is synonymous with
a non-vanishing Ricci tensor. 

\begin{theorem}\label{thm2:26oct12}
Let $g_0$ be a non-flat metric in $\cA^{s,p}$ such that $R(g_0) = 0$.  Then the linearization
$D_gR(g_0)$ is surjective and $\cS_2^{s,p} = \text{ker}(D_gR(g_0))\oplus \text{R}((D_gR(g_0))^*)$, where $(D_gR(g_0))^*$ is the
adjoint of $D_gR(g_0)$.  Moreover,
$R:\cA^{s,p} \to W^{s-2,p}$ maps any neighborhood of $g_0$ onto a neighborhood of $0$. 
\end{theorem}
\begin{proof}
See Theorem {\bf 1} in \cite{FiMa75}.
\end{proof}

\noindent
We now recall that if a $3$-dimensional compact manifold $\cM$ admits  a metric with positive scalar curvature, then
any $f \in C^{\infty}$ is the scalar curvature of some Riemannian metric $g$ on $\cM$ \cite{KW75,Au82}.  Therefore, for a given 
$\la \in \mathbb{R}$, the set of metrics $g$ on $\cM$ that satisfy $R(g) = \la$ will be non-empty.
Using this fact, Theorems~\ref{thm1:26oct12} and \ref{thm2:26oct12} and the Implicit Function Theorem~\ref{thm1:25apr13}, we can now prove the following theorem, which
allows us to conclude the existence of an analytic, one-parameter family of metrics $g_{\la}$ that satisfies $R(g_{\la})  = \la$.  

\begin{theorem}\label{thm3:26oct12}
Suppose that $\cM$ is a closed $3$-dimensional manifold that admits a metric with positive scalar curvature.
Then for $\la$ in a neighborhood of  $~0$, there exists an analytic one-parameter family of metrics $(g_{\la})$ through $g_0$ such that
$R(g_{\la}) = \la$.  
\end{theorem}
\begin{proof}
Because $\cM$ admits a metric with positive scalar curvature, it admits a non-flat metric $g_0$ with
zero scalar curvature.  Indeed, for some fixed $t_0 \in (0,1)$, one obtains the metric $g_0 = t_0h_0 + (1-t_0)h_1$ by taking a convex combination
of a metric $h_0$ with negative scalar curvature and a metric $h_1$ with positive scalar curvature.  In general, the
Ricci tensor of $g_0$ will be nonzero.  If it is zero, by fixing $h_0$ and perturbing $h_1$ to obtain $h_2 = h_1+\la h_3$,
where $h_3$ is non-flat metric that does not lie in the kernel of the linearized Ricci operator, one obtains the metric
$g_1 = t_1h_0+(1-t_1)h_2$ which has zero scalar curvature for some $t_1 \in (0,1)$ and will have a nontrivial
Ricci tensor for $ \la $ sufficiently small.  See \cite{Au82, KW75} for more details. 

Because $g_0$ is non-flat, Theorem~\ref{thm2:26oct12} implies $\cS^{s,p}_2 = \text{ker}(D_gR(g_0)) \oplus \text{R}((D_gR(g_0))^*)$.
Let $X = \text{R}((D_gR(g_0))^*) $ and define the operator
\begin{align}
&G: X\times \mathbb{R} \to W^{s-2,p},\\
&G(h,\la) = R(g_0+h ) - \la.\nonumber
\end{align}

Theorem~\ref{thm2:26oct12} and the splitting results in \cite{FiMa75} imply that for $h \in X$, $g_0+h$ determines an open subset of $\cS^{s,p}_2$.
Moreover, for $h$ sufficiently small, $g_0 +h \in \cA^{s,p}$ given that $\cA^{s,p}$ is an open subset of $\cS^{s,p}_2$.  Therefore,
there exists an open subset $U_1 \subset \cS^{s,p}_2$ about $g_0$ for which the scalar curvature operator is well-defined.
So for all $h \in X$ such that $g_0+h \in U_1$, $G(h,\la)$ is well-defined.  

By construction, $D_hG(0,0)$ is invertible and we may apply the Implicit Function Theorem in a neighborhood of $g_0$.  We conclude that there exists
a neighborhood $U_2\times V \subset U_1\times V \subset X \times \mathbb{R}$ of $(0,0)$ and a function $\psi: V \to U_2$, $\psi(0) = 0$, such that $G(h,\la) = 0$ in this neighborhood
if and only if $h = \psi(\la)$.  Letting $g_{\la} = g_0+\psi(\la) \in \cA^{s,p}$, we observe that $R(g_{\la}) - \la = G(\psi(\la),\la) = 0$, which implies that
$R(g_{\la}) = \la$ and $R(g_0) = 0$.  By Theorem~\ref{thm1:26oct12} and the Implicit Function Theorem~\ref{thm1:25apr13} the curve $g_{\la}$ is analytic in the variable $\la$. 
\end{proof}

\begin{remark}\label{rem1:26feb13}
The fact that $\psi(\la)$ is analytic in a neighborhood of $0$ means that for $\la$ sufficiently small,
\begin{align}
\lim_{N\to \infty}\|\psi(\la)-\sum_{i=0}^N\frac{1}{i!}D_{\la}^i\psi(0)\la^i\|_{W^{s,p}(T^0_2\cM)} = 0.  
\end{align}
Moreover, the sum 
\begin{align}\label{eq3:20may13}
\sum_{i=0}^{\infty} \frac{1}{i!}\|D_{\la}^i\psi(0)\||\la|^i
\end{align}
converges for $\la$ sufficiently small by Definition~\ref{def1:25apr13}, where $\|D_{\la}^i\psi(0)\|$ is the
operator norm \eqref{eq1:20may13} induced by the norm on $\mathbb{R}$ and the norm $\|\cdot\|_{W^{s,p}(T^0_2\cM)}$.
Therefore, if $1 \le k< s-\frac{3}{p}$, 
\begin{align}
\lim_{N\to \infty}\|\psi(\la)-\sum_{i=0}^N\frac{1}{i!}D_{\la}^i\psi(0)\la^i\|_{C^k(T^0_2\cM)} = 0.
\end{align}
This implies that if $g(x,\la) = g_0+\psi(\la)$, then in local coordinates 
\begin{align}\label{eq2:20may13}
\frac{\partial}{\partial x^m}(g_{ij}(x,\la)) = \sum_{i=0}^{\infty} \frac{1}{i!} \frac{\partial^{i+1}}{\partial^i\la \partial x^m}(g_{ij}(x,0)) \la^i,
\end{align}
for all $1 \le i,j,m \le 3$.  Furthermore, by \eqref{eq3:20may13} the series \eqref{eq2:20may13} converges absolutely. 
The same holds for higher order partials with respect to $x^m$ if $2 \le k \le s-\frac3p$.  See Proposition~\ref{prop1:24may13}
for further details.
\end{remark}

\section{Main Results}\label{sec1:1nov12}  

Let $\cM$ be a closed, $3$-dimensional manifold which admits a metric with positive scalar curvature that also admits a non-flat metric $g_0\in \cA^{s,p}$ such that $R(g_0) = 0$.
Let $(g_{\la})$ be the analytic curve of metrics determined in Theorem~\ref{thm3:26oct12}.
Define the operator
\begin{align}\label{eq2:11mar13}
F((\phi,\bw),\la) = \left[\begin{array}{c} 
- \Delta_{\la}\phi
+ \frac{1}{8}\la \phi
+ \frac{\la^4}{12}\tau^2\phi^5
-a_{\bw,\lambda}
\phi^{-7}
-\frac{\la^2\kappa}{4}\rho\phi^{-3}
\\
 \mathbb{L}_{\la}\bw
+\frac{2\la^2}{3}D_{\la}\tau\phi^6
+\la^2\kappa\bj
\end{array}\right],
\end{align}
where 
$a_{\bw,\lambda} = \frac{1}{8}(\la^2\sigma+\mathcal{L}{\bw})_{ab}(\la^2\sigma+\mathcal{L}\bw)^{ab}$,
and
where $\Delta_{\la}, D_{\la}$ and $\mathbb{L}_{\la}$ are induced by $(g_{\la})$.
We view $F((\phi,\bw),\la)$ as a nonlinear operator, where
\begin{align}
F((\phi,\bw),\la) : C^{2,\alpha}\oplus C^{2,\alpha}(T\cM)\oplus \mathbb{R} \to C^{0,\alpha}\oplus C^{0,\alpha}(T\cM),
\end{align}
and if $F((\phi_0,\bw_0),\la_0) = (0,{\bf 0})$, then $(\phi_0,\bw_0)$ solves \eqref{eq1:26oct12} when $\la = \la_0$. 

Clearly we have that $F((1,{\bf 0}),0) = 0$.  Moreover, we will show that $\text{ker}D_XF((1,{\bf 0}),0)$ is one-dimensional.  We
can then use Theorem~\ref{thm1:29apr12} to parametrize a solution curve $((\phi(s),\bw(s)),\la(s))$ through $((1,{\bf 0}),0)$.
The first of our two main results
in this paper characterizes the behavior of solutions on this curve in a neighborhood of $((1,{\bf 0}),0)$.

\begin{theorem}\label{thm1:31oct12}
Let $\cM$ be a closed 3-dimensional manifold that admits an analytic, one-parameter family of metrics $g_{\la}\subset \cA^{s,p}$, $s > 3+\frac3p$,
such that for each $\la \in (-\delta,\delta)$, $R(g_{\la}) = \la$ and $g_{\la}$ has no conformal Killing fields.  
Suppose that $(\tau,\sigma,\rho,\bj) \in C^1(\cM)\times C(\cM)\times C(\cM) \times C(T\cM)$ is freely specified, and using this data and
the one-parameter family $g_{\la}$, define $F((\phi,\bw),\la)$ as in \eqref{eq2:11mar13}.  Then at least one of the following two possibilities must occur:
\begin{enumerate}
\item There exists a $\delta_0 \in (0,\delta)$ such that for all $\la \in (0,\delta_0)$ there exists $(\phi_{1,\la},\bw_{1,\la})$ and $(\phi_{2,\la},\bw_{2,\la})$ in $C^{2,\alpha}\oplus C^{2,\alpha}(T\cM)$ 
that together solve \eqref{eq1:26oct12} with $(\phi_{1,\la_0},\bw_{1,\la_0}) \ne (\phi_{2,\la_0},\bw_{2,\la_0})$,
\item There exists a $\delta_0 \in (0,\delta)$ such that for any $\la \in (-\delta_0, 0)$, there exists $(\phi_{\la},\bw_{\la}) \in C^{2,\alpha}\oplus C^{2,\alpha}(T\cM)$ that 
solves \eqref{eq2:11mar13}. 
\end{enumerate}
\end{theorem}

Combining Theorem~\ref{thm2:26oct12} and Theorem~\ref{thm1:31oct12}, we obtain our second main result. 

\begin{theorem}\label{thm2:31oct12}
Let $\cM$ be a closed 3-dimensional manifold which admits both a metric with positive scalar curvature
and a metric $g_0$ with zero scalar curvature and no conformal Killing fields, where both metrics are contained
in $\cA^{s,p}$, $s>3+\frac3p$.  
Let $(\tau,\sigma,\rho,\bj) \in C^1(\cM)\times C(\cM)\times C(\cM) \times C(T\cM)$ be freely specified data for the 
CTT formulation of the constraints \eqref{eq1:26oct12}.   
Then in any neighborhood $U$ of $g_0$ there
exists a metric $g \in \cA^{s,p}$ and a $\la>0$ such that at least one the following must hold:
\begin{itemize}
\item $R(g) = \la$ and solutions to the CTT formulation of the Einstein Constraints with specified data ${(g, \la^2\tau, \la^2\sigma, \la^2\rho, \la^2\bj)}$
are non-unique

\item $R(g) = -\la$ and there exists a solution to CTT formulation of the Einstein Constraints with specified data ${(g,\la^2\tau, \la^2\sigma, \la^2\rho,\la^2\bj)}$.
\end{itemize}
Thus, in any neighborhood of a metric with zero scalar curvature and no conformal Killing fields, either there exists a Yamabe 
positive metric for which solutions to the CTT formulation are non-unique or
there exists a Yamabe negative metric for which far-from-CMC solutions to the CTT formulation exist. 
\end{theorem}

\begin{remark}\label{rem1:5mar13}
An important point of Theorem~\ref{thm2:31oct12} is that the function $\tau$ is an arbitrary, continuously differentiable function.
Therefore this function is allowed to have zeroes and is free of any near-CMC conditions.
\end{remark}

\begin{remark}\label{rem1:22apr13}
Here we do not prove the existence of manifolds $\cM$ that admit both a metric of positive scalar curvature
and a metric with zero scalar curvature and no conformal Killing fields.  Similar assumptions are
made in \cite{ACI08}, and using the results in \cite{BPS05,B76,EIM86}, we can conclude that using a suitable topology,
the set of metrics on a given manifold $\cM$ which
have no homothetic Killing fields is generic in the set of metrics with zero scalar curvature.  More generally, the set
of metrics with no conformal Killing fields is a generic set in the space of metrics on $\cM$ \cite{BPS05}.  We suspect that
these results can be used to show,
under possibly additional regularity assumptions,
that manifolds which admit both a metric of positive scalar curvature and a metric 
with zero scalar curvature and no conformal Killing vectors exist. 
\end{remark}

\section{Properties of $F((\phi,\bw),\la)$}\label{analytic}

In this section we discuss some key properties of the operator $F((\phi,\bw),\la)$ introduced in \eqref{eq2:11mar13}.  
Our general strategy to prove the main results in Section~\ref{sec1:1nov12} 
will be to apply a Liapunov-Schmidt reduction to this operator.
In order to apply this reduction, we seek a point $((\phi_0,\bw_0),\la_0)$ for which
the linearization $D_XF((\phi_0,\bw_0),\la_0)$ has a nontrivial kernel, where $X= (\phi,\bw)$.  

In the following discussion, we assume that $\cM$ is a closed, $3$-dimensional manifold that admits an analytic, one-parameter
family of metrics satisfying $R(g_{\la}) = \la$ for $\la \in (-\delta,\delta)$.  Additionally assume that each $g_{\la}$ has no conformal Killing fields
and $(g_{\la}) \subset \cA^{s,p}$, where $s > 3+\frac3p$. 
Assuming that $(\tau, \sigma, \rho,\bj)$ is given
data for the conformal formulation, we may define the operator $F((\phi,\bw),\la)$ as in \eqref{eq2:11mar13} and we
have the following result:

\begin{proposition}\label{prop1:2nov12}
Let $F((\phi,\bw),\la)$ be the nonlinear operator defined in \eqref{eq1:1nov12}.  Then the following holds:
\begin{align}\label{eq2:2nov12}
D_XF((1,{\bf 0}),0) = \left[ \begin{array}{cc} -\Delta & 0 \\ 0 & \mathbb{L} \end{array}\right] \quad \text{and} \quad \text{ker}(D_XF((1,{\bf 0}),0)) = \text{span}\left\{ \left[\begin{array}{c} 1\\{\bf 0}\end{array}\right]\right\} ,
\end{align}
where $\Delta$ and $\mathbb{L}$ are the Laplace-Beltrami operator and the negative divergence of the conformal Killing operator induced by $g_0$. 
\end{proposition}
\begin{proof}
This follows from the fact that the Gauteaux derivative and Frechet derivative coincide in a neighborhood of $((1,{\bf 0}),0)$.  Therefore,
for $(\phi,\bw)$ satisfying
$$\|(\phi,\bw)\|_{C^{2,\alpha}(\cM)\oplus C^{2,\alpha}(T\cM)}= 1,$$ 
we compute 
$$
\lim_{t \to 0} \frac{F((1,{\bf 0}) + t(\phi,\bw)),0) - F((1,{\bf 0}),0)}{t}
$$
to obtain \eqref{eq2:2nov12}.  Given that $g_0$ has no conformal Killing fields,
it is clear that the kernel of \eqref{eq2:2nov12} is spanned by  $\left[\begin{array}{c} 1\\{\bf 0}\end{array}\right]$.
\end{proof}

\begin{remark}\label{rem1:11mar13}
Clearly the operator $D_XF((1,0),{\bf 0})$ is a self-adjoint operator.  Therefore, Proposition~\ref{prop1:2nov12} also implies that
$\text{ker}((D_XF(1,{\bf 0}),0)^*)  = \left[\begin{array}{c} 1\\{\bf 0}\end{array}\right]$.
\end{remark}

We will also require that the operator $F((\phi,\bw),\la)$ have certain regularity properties in a neighborhood of the point
$((1,{\bf 0}),0)$.  For this we have the following proposition:
\begin{proposition}\label{prop2:2nov12}
In a neighborhood of $((1,{\bf 0}),0)$, the nonlinear operator $F((\phi,\bw),\la)$ is an analytic operator
between the spaces
$$
C^{2,\alpha}\oplus C^{2,\alpha}(T\cM) \oplus \mathbb{R} \to C^{0,\alpha} \oplus C^{0,\alpha}(T\cM).
$$
\end{proposition}
 
\begin{proof}

Writing out $F((\phi,\bw),\la)$ on a given chart element $U_j$, the Hamiltonian constraint, which we will denote by $F_1((\phi,\bw),\la)$, takes the form 
\begin{align}\label{eq1:26feb13}
F_1(( & \phi,\bw),\la) = \\
&f_1^{ab}(\la)\partial_a\partial_b\phi+f_2^a(\la)\partial_a\phi+\frac18\la\phi+\frac{\la^4}{12}\tau^2\phi^5-\frac{\phi^{-7}}{8}\left(f_3^{abcd}(\la)\partial_aw_b\partial_cw_d \right. +\nonumber \\
&\left.f_4^{abc}(\la)\partial_aw_bw_c+f^{ab}_5(\la)w_aw_b+\la^2f_6^{ab}\partial_aw_b+\la^2f_7^a(\la)w_a+\la^4\sigma^2\right)-\frac{\kappa\la^2}{4}\rho\phi^{-3},\nonumber
\end{align}
where $f^{ab}_1,...,f^{a}_7$ are functions in $C^{1,\alpha}(U_j \times (-\delta,\delta))$, $\alpha = 1+[\frac 3p]-\frac 3p$, 
that are formed from sums and products of the first and second derivatives of the components of $g_{\la}$ with respect to the spatial coordinate functions $x^i$.  See Proposition~\ref{prop1:26feb13} for details.  Given that
the $g_{\la}$ are analytic in $\la \in (-\delta,\delta)$, Remark~\ref{rem1:26feb13} and Proposition~\ref{prop1:24may13} imply that these functions are also analytic for $\la \in (-\delta,\delta)$.  Similarly,
the momentum constraint $F_2((\phi,\bw),\la)$ takes the form 
\begin{align}\label{eq2:26feb13}
F_2((\phi,\bw),&\la) = \\
&h_1^{abcd}(\la)\partial_a\partial_bw_c+h_2^{abd}(\la)\partial_aw_b+h^{ad}_3(\la)w_a+\frac23\la^2h_4^{ad}(\la)\partial_a\tau\phi^6+\la^2\kappa j^d, \nonumber
\end{align} 
where $h^{abcd}_1,..., h_4^{ad} \in C^{1,\alpha}(U_j\times (-\delta,\delta))$ and are analytic with respect to $\la \in (-\delta,\delta)$.  See Proposition~\ref{prop2:26feb13} for further discussion. 

Expanding $f^{ab}_1,...,f^{a}_7$ about $\la = 0$ and $(\phi+1)^5, (\phi+1)^{-7}, (\phi+1)^{-3}$ about $\phi = 0$, we obtain the following power series
representation for the Hamiltonian constraint for $((\phi,\bw),\la)$ in a neighborhood of $((1,{\bf 0}),0)$:
\begin{align}\label{eq1:2mar13}
F_1((&\phi+1,\bw),\la)= \\
&\sum_{i=0}^{\infty} \frac{1}{i!}\frac{\partial^if_1^{ab}(0)}{\partial \la^i} \la^i\partial_a\partial_b\phi + 
 \sum_{i=0}^{\infty} \frac{1}{i!}\frac{\partial^if_2^{a}(0)}{\partial \la^i} \la^i\partial_a\phi + \frac18\la(\phi+1) +\sum_{i=0}^5\frac{\tau^2}{12}\binom{5}{i}\phi^i\la^4 \nonumber\\  
+& \sum_{i=0}^{\infty} \frac{(-1)^{i+1}(i+2)!}{8(i!)}\kappa\rho\phi^i\la^2+ \sum_{i,j=0}^{\infty} \frac{(-1)^{i+1}(i+6)!}{8(6!)(i!)(j!)}\frac{\partial^jf_3^{abcd}(0)}{\partial \la^j}\phi^i \la^j(\partial_aw_b)(\partial_cw_d)    \nonumber  \\
+&\sum_{i,j=0}^{\infty} \frac{(-1)^{i+1}(i+6)!}{8(6!)(i!)(j!)}\frac{\partial^jf_4^{abc}(0)}{\partial \la^j}\phi^i \la^j(\partial_aw_b)w_c \nonumber \\
+&\sum_{i,j=0}^{\infty} \frac{(-1)^{i+1}(i+6)!}{8(6!)(i!)(j!)}\frac{\partial^jf_5^{ab}(0)}{\partial \la^j}\phi^i \la^j(w_a)(w_b) \nonumber\\
+&\sum_{i,j=0}^{\infty} \frac{(-1)^{i+1}(i+6)!}{8(6!)(i!)(j!)}\frac{\partial^jf_6^{ab}(0)}{\partial \la^j}\phi^i \la^{j+2}(\partial_aw_b) \nonumber \\
+&\sum_{i,j=0}^{\infty} \frac{(-1)^{i+1}(i+6)!}{8(6!)(i!)(j!)}\frac{\partial^jf_7^{a}(0)}{\partial \la^{j}}\phi^i \la^{j+2}(w_a) \nonumber \\
+& \sum_{i=0}^{\infty} \frac{(-1)^{i+1}(i+6)!}{8(6!)(i!)}\phi^i \la^4\sigma^2. \nonumber
\end{align}  
Similarly, by expanding out $h_1^{abcd}, h_2^{abd}, h_3^{ad}, h_4^{ad}$ with respect to $\la$ about $\la = 0$ and ${(\phi+1)^6}$ about $\phi = 0$, we obtain a power series representation of
the momentum constraint for $((\phi,\bw),\la)$ in a neighborhood of $((1,{\bf 0}),0)$:
\begin{align}\label{eq2:2mar13}
F_2((&\phi+1,\bw),\la)= \\
& \sum_{i=0}^{\infty}\frac{1}{(i!)}\frac{\partial^ih_1^{abcd}(0)}{\partial \la^i}\la^i\partial_a(\partial_bw_c)+  
\sum_{i=0}^{\infty}\frac{1}{(i!)}\frac{\partial^ih_2^{abd}(0)}{\partial \la^i}\la^i(\partial_aw_b) \nonumber \\
+& \sum_{i=0}^{\infty}\frac{1}{(i!)}\frac{\partial^ih_3^{ad}(0)}{\partial \la^i}\la^i(w_a)
 + \sum_{i=0}^{\infty}\sum_{j=0}^6\frac{2}{3(i!)}\binom{6}{j}\frac{\partial^ih_4^{ad}(0)}{\partial \la^i}\partial_a\tau \phi^j\la^{i+2}  + \la^2\kappa j^d. \nonumber  
\end{align}
The regularity of the coefficients $f_1,\cdots, f_7$, Proposition~\ref{prop1:24may13}, Remark~\ref{rem1:26feb13} and the fact that $\phi \in C^{2,\alpha}$ imply that the series in \eqref{eq1:2mar13} 
converges to $F_1((\phi,\bw),\la)$ in $C^{0,\alpha}(\cM)$ for $|\phi| <1$ and $|\la| < \delta$.  Similarly, 
the series in \eqref{eq2:2mar13} converges to $F_2((\phi,\bw),\la)$ in $C^{0,\alpha}(T\cM)$ for $|\phi|<1$ and $|\la|<\delta$.

Let $\bh = ((\phi,\bw),\la)$, $\bx_0 = ((1,{\bf 0}),0)$.  We can rewrite the power series representations of $F_1$ and $F_2$ in Eqs. \eqref{eq1:2mar13} and \eqref{eq2:2mar13} to express $F((\phi+1,\bw),\la) = F(\bx_0+\bh)$ as
a power series of multilinear operators.  For a given multi-index $\alpha = (\alpha_1,\alpha_2,\alpha_3)$, $|\alpha| = k$, define $D^{\alpha}F_i(\bx_0+h)|_{h=0}$ to be the resulting
operator obtained by partially differentiating the power series representations of $F_1(\bx_0+\bh)$ and $F_2(\bx_0+\bh)$ with respect to the multi-index $\alpha = (\alpha_1,\alpha_2,\alpha_3)$, where
we differentiate $\alpha_1$ times with respect to $\phi$, $\alpha_2$ times with respect to $\bw$, and $\alpha_3$ times with respect to $\la$.  
Here $D^{\alpha}F_i(\bx_0+\bh)|_{h=0}$ is an $\alpha_1$-multilinear operator on $C^{2,\alpha}$, an $\alpha_2$-multilinear operator on $C^{2,\alpha}(T\cM)$, and
an $\alpha_3$-multilinear operator on $\mathbb{R}$.  Then by a slight abuse of notation, we may succinctly write
\begin{align}\label{eq1:29apr13}
M_{i,\alpha}({\bf x}_0)h^{\alpha} = \text{\hspace{.75 in}} &\\
D^{\alpha}F_i(\bx_0+h)|_{h=0}&(\underbrace{\phi,\cdots,\phi},\underbrace{\bw,\cdots,\bw},\underbrace{\la,\cdots,\la}), \quad  \text{for $i =1,2$}.\nonumber \\
&\text{\hspace{.2 in}} \alpha_1 ~~ \text{times}  \text{\hspace{.17 in}}   \alpha_2 ~~\text{times} \text{\hspace{.16 in} } \alpha_3~~\text{times} \nonumber
\end{align}
We then define a $k$-linear operator for $\bh \in C^{2,\alpha}\times C^{2,\alpha}(T\cM) \times \mathbb{R}$ by letting
\begin{align}
M_k(\bx_0)\bh^k = \left[\begin{array}{c}\sum_{\alpha: ~|\alpha|= k} \frac{k!}{(\alpha_1!)( \alpha_2!)( \alpha_3!)}M_{1,\alpha}h^{\alpha}\\ \sum_{\alpha:~|\alpha| = k}\frac{k!}{(\alpha_1!)( \alpha_2!)( \alpha_3!)}M_{2,\alpha}h^{\alpha} 
\end{array}\right] \in C^{0,\alpha}\times C^{0,\alpha}(T\cM),
\end{align}
where the sums are over all three-tuples $(\alpha_1,\alpha_2,\alpha_3)$ such that $\alpha_i \ge 0$.

Then by Eqs.~\eqref{eq1:2mar13}-\eqref{eq2:2mar13} we have that on each chart element $U_j$,
\begin{align}\label{eq2:29apr13}
F(\bx_0+ \bh,\la) = \left[\begin{array}{c} F_1(\bx_0+\bh,\la) \\ F_2(\bx_0+\bh,\la)\end{array}\right] = \sum_{k=1}^{\infty} M_k(\bx_0)\bh^k.
\end{align}
This follows since the expression $M_k({\bf x}_0)h^k$ is obtained by grouping all terms of combined order $k$ in $\phi, \bw$ and $\la$ in Eqs. \eqref{eq1:2mar13}-\eqref{eq2:2mar13}.  
We may rearrange the series representations of $F_1(\bx_0+\bh)$ and $F_2(\bx_0+\bh)$ given that the series in Eqs. \eqref{eq1:2mar13}-\eqref{eq2:2mar13} converge absolutely in the
sense of \eqref{eq2:25apr13} for $|\la| < \delta$ and the power series expansions involving $(\phi+1)^{-7}, (\phi+1)^{-3}$ converge uniformly for $|\phi|<1$.  See Proposition~\ref{prop1:24may13}
for details.  By the same reasoning, we also have that on each $U_j$ the series representation \eqref{eq2:29apr13} will converge absolutely in the sense of \eqref{eq2:25apr13}.
By a partition of unity argument, we can conclude that the operator $F((\phi,\bw),\la)$ is an analytic operator if $|\phi|<1$ and $|\la |< \delta$.
\end{proof}

\section{Proof of Main Results}\label{mainproof}

In this section we will parametrize solutions to $F((\phi,\bw),\la) = 0$ in a neighborhood of $((1,{\bf 0}),0)$, where we recall that
\begin{align}\label{eq1:1nov12}
F((\phi,\bw),\la) =  \left[\begin{array}{c}
- \Delta_{\la}\phi
+ \frac{1}{8}\la \phi
+ \frac{\la^4}{12}\tau^2\phi^5
- a_{\bw,\lambda}
 \phi^{-7}
-\frac{\la^2\kappa}{4}\rho\phi^{-3}\\
  \mathbb{L}_{\la}\bw
+ \frac{2\la^2}{3}D_{\la}\tau\phi^6
+ \la^2\kappa\bj \end{array}\right],
\end{align}
where 
$a_{\bw,\lambda} = \frac{1}{8}(\la^2\sigma+\mathcal{L}{\bw})_{ab}(\la^2\sigma+\mathcal{L}\bw)^{ab}$,
and where
$(\tau,\sigma,\rho,\bj) \in C^1(\cM)\times C(\cM)\times C(\cM) \times C(T\cM)$ is specified data and 
$g_{\la}$ is a one-parameter family of metrics defining the operators $\Delta_{\la}, \mathbb{L}_{\la}$ and $D_{\la}$.
Our approach is to apply the Liapunov-Schmidt reduction in Section~\ref{LSchmidt} to \eqref{eq1:1nov12} to determine
an explicit solution curve through the point $((1,{\bf 0}),0)$.  The analyticity of $F((\phi,\bw),\la)$ and $g_{\la}$ will imply that
this solution curve is analytic in its parametrizing variable.  This result along with
the preexisting far-from-CMC solution theory established in \cite{HNT08,HNT09} will imply the results in Section~\ref{sec1:1nov12}. 

\begin{proof}[Proof of Theorem~\ref{thm1:31oct12}]
Let $g_{\la}$ be the one-parameter family of metrics defined in Theorem~\ref{thm1:31oct12}.  
Given data $(\tau,\sigma,\rho,\bj) \in C^1(\cM)\times C(\cM)\times C(\cM) \times C(T\cM)$ for the conformal equations,
we then define an associated one-parameter family of nonlinear operators ${F((\phi,\bw),\la)}$ as in \eqref{eq1:1nov12}.
By Proposition~\ref{prop1:2nov12} we know that $\text{ker}D_XF((1,{\bf 0}),0)$
takes the form
$$
D_XF(1,{\bf 0},0) =  \left[ \begin{array}{cc} -\Delta &  0\\
0 & \mathbb{L}\end{array} \right],
$$
and that 
$\text{ker}(D_XF((1,{\bf 0}),0))$ and $\text{ker}(D_XF((1,{\bf 0}),0)^*)$ are spanned by $\hat{v}_0 =\tiny{\left[\begin{array}{c}1\\ 0 \end{array}\right]}$.

We decompose  
$$
X = C^{2,\alpha}(\cM)\oplus C^{2,\alpha}(T \cM) = X_1 \oplus X_2 ,
$$
and
$$
Y = C^{0,\alpha}(\cM)\oplus C^{0,\alpha}(T \cM) = Y_1\oplus Y_2,
$$ 
where
\begin{align}\label{eq5:31july12}
&X_1 = \text{ker}(D_XF((1,{\bf 0}),0)), \\
&X_2 =  R(D_XF((1,{\bf 0}),0)^*)\cap (C^{2,\alpha}(\cM)\oplus C^{2,\alpha}(T\cM)),\\
&Y_1 = R(D_XF((1,{\bf 0}),0))\cap (C^{0,\alpha}(\cM)\oplus C^{0,\alpha}(T\cM)), \hspace{2mm}\\ 
&Y_2 = \text{ker}(D_XF((1,{\bf 0}),0)^*).
\end{align}
For justification that we can decompose $X$ and $Y$ in the manner described above,
see the appendix of \cite{CMH12}.
\medskip

Let $P:X\to X_1$ and $Q:Y\to Y_2$ be projection operators defined using $\hat{v}_0$. 
Then by writing 
$$\left[ \begin{array}{c}\phi\\ \bw\end{array}\right]= P\left[ \begin{array}{c}\phi\\ \bw\end{array}\right]+ (I-P)\left[ \begin{array}{c}\phi\\ \bw\end{array}\right] = v+y,$$ 
where $v\in X_1$ and $y \in X_2$, the Implicit Function Theorem~\ref{thm1:25apr13} applied to
\begin{align}\label{eq2:22aug12}
(I-Q)F(v+y,\la) = 0,
\end{align}
implies that solutions to $F((\phi,\bw),\la)=0$ satisfy
\begin{align}\label{eq2:19july12}
\Phi(v,\la) = QF(v+\psi(v,\la),\la) = 0,
\end{align}
in a neighborhood of $((1,{\bf 0}),0)$, where $y = \psi(v,\la)$ in this neighborhood and where ${(0,{\bf 0}) = \psi((1,{\bf 0}),0)}$.

By Proposition~\ref{prop2:2nov12} and Theorem~\ref{thm1:25apr13} the curve $\psi(v,\la)$ 
is analytic in $v$ and $\la$.
Furthermore, 
$$D_{\la}F((1,{\bf 0}),0) = \left[\begin{array}{c} 1/8 \\ 0 \end{array}\right]  \in X_1.$$ 
Therefore, $D_{\la}F((1,{\bf 0}),0)\notin R(D_XF((1,0),0))$, and we can apply Theorem~\ref{thm1:29apr12}
to conclude there exists a $\delta >0$ such that all
solutions to $F((\phi,\bw),\la) = 0$ in a neighborhood of $((1,{\bf 0}),0)$ are parametrized by 
$s \in (-\delta,\delta)$ in the following way: 
\begin{align}\label{eq4:2nov12}
&(\phi(s),\bw(s)) = \hat{v}_0 + s\hat{v}_0 + \psi(\hat{v}_0+s\hat{v}_0, \gamma(\hat{v}_0+s\hat{v}_0)),\\ 
&\la(s) = \gamma(s\hat{v}_0+\hat{v}_0) \nonumber.
\end{align}
In \eqref{eq4:2nov12}, $\gamma: U \subset X_1 \to (-\ee, \ee) \subset \mathbb{R}$ is analytic in a neighborhood of $(1,{\bf 0})$, 
and is obtained by applying the Implicit Function Theorem~\ref{thm1:25apr13}
to the operator $QF(v+\psi(v,\la),\la)$, which is analytic in a neighborhood of $((1,{\bf 0}),0)$.  We write $v = (s+1)\hat{v}_0$ given
that $X_1$ is $1$-dimensional.

Now we observe that if we choose $\la$ sufficiently small so that the
size conditions in the positive Yamabe far-from-CMC results in Theorem~\ref{farCMC} are
satisfied, then for any $\la>0$ sufficiently small, solutions to $F((\phi,\bw),\la) = 0$ will exist. 
Therefore, after possibly shrinking the intervals $(-\delta,\delta)$ and $(-\ee,\ee)$, there must exist an $s \in (-\delta,\delta)$ such that ${\la(s) = \gamma(s \hat{v}_0 + \hat{v}_0) = \la}$ for each $\la \in (0,\ee)$.  Now we 
summarize the properties of the function $\la(s)$.
\begin{itemize}
\item $\la(s)$ is analytic on the interval $(-\delta,\delta)$.

\item For any $\la \in (0,\ee)$, there exists an $s \in (-\delta,\delta)$ so that $\la(s) = \la$.

\item $\la(0) = 0$.
\end{itemize} 

The first two properties tell us that the interval $s \in (-\delta,\delta)$ cannot contain a set of zeros of $\la(s)$
with a limit point in $(-\delta,\delta)$.  In particular, we conclude that $\la(s)$ cannot vanish on any subinterval $I \subset (-\delta,\delta)$.
Therefore, one of following two possibilities must occur:
\medskip
\begin{enumerate}
\item There exists $\la \in (0,\ee)$ and $s_1,s_2 \in (-\delta,\delta)$,  $s_1 \ne s_2$, such that \\
$\la(s_1) =\la(s_2) = \la$.\label{eq1:5mar13}\\
\item There exists $\la \in (-\ee,0)$ and $s_0 \in (-\delta,\delta)$ such that $\la(s_0) = \la$.\label{eq2:5mar13}
\end{enumerate}
\smallskip

If \eqref{eq2:5mar13} occurs, then
\begin{align}\label{eq3:5nov13}
&(\phi_0,\bw_0) = (\phi(s_0),\bw(s_0)) = \hat{v}_0 + s_0\hat{v}_0 + \psi(\hat{v}_0+s_0\hat{v}_0, \gamma(\hat{v}_0+s_0\hat{v}_0)),\\ 
&\la_0 = \la(s_0) = \gamma(s_0\hat{v}_0+\hat{v}_0)<0,
\end{align}
satisfies $F((\phi_0,\bw_0),\la_0) = 0$.  This implies that the data set $(g_{\la_0}, \la_0^2\tau, \la_0^2\sigma, \la_0^2\rho, \la_0^2\bj)$
yields the solution $(\phi_0,\bw_0)$ to the conformal equations.

If \eqref{eq1:5mar13} holds, then both
 \begin{align}\label{eq4:5nov13}
&(\phi_i,\bw_i) = (\phi(s_i),\bw(s_i)) = \hat{v}_0 + s_i\hat{v}_0 + \psi(\hat{v}_0+s_i\hat{v}_0, \gamma(\hat{v}_0+s_i\hat{v}_0)),\\ 
&\la_i = \la(s_i) = \gamma(s_i\hat{v}_0+\hat{v}_0) = \la,
\end{align}
satisfy $F((\phi_i,\bw_i),\la) = 0$ for $i \in \{1,2\}$.  We showed in \cite{CMH12} that the operator 
$$f(s) = \psi(\hat{v}_0+s\hat{v}_0, \gamma(\hat{v}_0+s\hat{v}_0))= \cO(s^2) \quad \text{as} ~~ s\to 0.$$
The argument there followed by differentiating $f(s)$ with respect to $s$ and showing that $\dot{f}(0) = 0$,
which we can conclude from Proposition~\ref{prop1:6july12}.  This fact ensures that for $s$ in a small neighborhood
of $0$, the solutions $(\phi_1,\bw_1)$ and $(\phi_2,\bw_2)$ will be distinct.  This completes the proof of Theorem~\ref{thm1:31oct12}.

\end{proof}

\begin{proof}[Proof of Theorem~\ref{thm2:31oct12}]
If $\cM$ admits a metric with positive scalar curvature and a scalar flat metric $g_0$ with no conformal Killing fields,
we can apply Theorem~\ref{thm3:26oct12} to conclude that there exists a one-parameter family of metrics $g_{\la}$
through $g_0$ such that $R(g_{\la}) = \la$.  Moreover, since the set of metrics with no conformal Killing fields is an open
dense set, for $\la $ sufficiently small the metrics $g_{\la}$ will have no conformal Killing fields.  See \cite{BPS05} for
details.  We can therefore apply Theorem~\ref{thm1:31oct12} to conclude our result.
\end{proof}

\section{Conclusion}
\label{conc}
For a given closed, $3$-dimensional manifold $\cM$ that admits a metric with positive scalar curvature
we showed in Section~\ref{curvature} that there exists an analytic, one-parameter family
of metrics $g_{\la}$ that satisfies $R(g_{\la}) = \la$.
By adding the extra assumption that $\cM$ also admitted a metric
$g_0$ with zero scalar curvature and no conformal Killing fields, we were able to obtain an analytic family $g_{\la}$ through
$g_0$ with no conformal Killing fields that satisfied $R(g_{\la}) = \la$.  Using this one-parameter family and
given data $(\tau,\sigma,\rho,\bj)$ for the conformal equations, in Section~\ref{sec1:1nov12} we constructed a nonlinear operator 
\begin{align}\label{eq3:11mar13}
F((\phi,\bw),\la) = \left[\begin{array}{c}
- \Delta_{\la}\phi
+ \frac{1}{8}\la \phi
+ \frac{\la^4}{12}\tau^2\phi^5
- a_{\bw,\lambda}
  \phi^{-7}
- \frac{\la^2\kappa}{4}\rho\phi^{-3}\\
  \mathbb{L}_{\la}\bw
+ \frac{2\la^2}{3}D_{\la}\tau\phi^6
+ \la^2\kappa\bj \end{array}\right],
\end{align}
with 
$a_{\bw,\lambda} = \frac{1}{8}(\la^2\sigma+\mathcal{L}{\bw})_{ab}(\la^2\sigma+\mathcal{L}\bw)^{ab}$,
where solutions to $F((\phi,\bw),\la) =0$ satisfy the conformal equations with given data $(g_{\la}, \la^2\tau, \la^2\sigma, \la^2\rho,\la^2\bj)$. 
In Section~\ref{analytic}, we then showed that the nonlinear operator \eqref{eq3:11mar13} was analytic, and
in section~\ref{mainproof} we parametrized solutions to the
nonlinear problem
${F((\phi,\bw),\la) = 0}$ in a neighborhood of $((1,{\bf 0}),0)$. 

The analyticity of $F((\phi,\bw),\la)$ implied that our parametrized solution curve
\begin{align}\label{eq4:11mar13}
&(\phi(s),\bw(s)) = \hat{v}_0 + s\hat{v}_0 + \psi(\hat{v}_0+s\hat{v}_0, \gamma(\hat{v}_0+s\hat{v}_0)),\\ 
&\la(s) = \gamma(s\hat{v}_0+\hat{v}_0),
\end{align}
was analytic for $s\in (-\delta,\delta)$.  Using the analyticity of the solution curve~\eqref{eq4:11mar13} and
the preexisting far-from-CMC solution theory from
\cite{HNT08,HNT09}, we were then able to conclude that one of the following two must possibilities must hold:
\begin{enumerate}
\item There exists $\la_0 \in (0,\ee)$ and $s_1,s_2 \in (-\delta,\delta)$,  $s_1 \ne s_2$, such that \\
$(\phi(s_1),\bw(s_1)) \ne (\phi(s_2),\bw(s_2)), ~~\la(s_1) =\la(s_2) = \la_0$.\\
\item There exists $\la_0 \in (-\ee,0)$ and $s_0 \in (-\delta,\delta)$ such that $\la(s_0) = \la_0$.
\end{enumerate}
These two possibilities and Theorem~\ref{thm3:26oct12} implied the conclusions of Theorem~\ref{thm1:31oct12} and Theorem~\ref{thm2:31oct12}, the two main results of
our paper.  Namely, we concluded that either the positive Yamabe, far-from-CMC solutions to the constraint
equations must be non-unique, or that negative Yamabe, far-from-CMC solutions exist for this class of manifolds.  

While this article does not provide specific criteria for when positive Yamabe, far-from CMC solutions are non-unique and when
negative Yamabe, far-from-CMC solutions exist, it does show that one of these two possibilities must hold for this manifold class.  
Given that both of these aspects of the far-from-CMC solution theory are completely unresolved, these results further extend our understanding of the conformal method, and also provide some new analytical tools for obtaining additional results in this direction.  In an effort to push this line of research further,
we are currently working on a concrete way to distinguish between the cases above.  Our analysis lies in
whether the first non-zero term in the Taylor expansion of $\la(s)$ even or odd.  That is, if $\la(s)$ is of the form 
$$\la(s) = \frac{d^{(2i+1)}\la}{d\la^{2i+1}}(0)s^{2i+1}+\cO(s^{2i+2})\quad \text{for $i \ge 2$},$$
then negative Yamabe, far-from-CMC solutions exist for this class of metrics.  On the other hand, if
$$\la(s) = \frac{d^{(2i)}\la}{d\la^{2i}}(0)s^{2i}+\cO(s^{2i+1})\quad \text{for $i \ge 2$},$$
then the positive Yamabe, far-from-CMC solutions determined in \cite{HNT08,HNT09} are non-unique.  In order to determine
which form $\la(s)$ has, one needs to express $\frac{d^i}{d\la^i}\la(0)$ in terms of higher order derivatives of $F((\phi,\bw),\la)$ as in Proposition~\ref{prop1:1july12}.
This research is currently under way.

Another interesting oberservation that can be made from our results is that in Theorems~\ref{thm1:31oct12}-\ref{thm2:31oct12}, no distinction is made between
the near-CMC and far-from-CMC cases.  We simply don't assume that the near-CMC conditions hold.  Given that solutions to the conformal equations
are unique in the near-CMC case, we must have that solutions to the nonlinear problem $F((\phi,\bw),\la) = 0$ are unique in the event that the specified data
$\tau$ satisfies the near-CMC assumption.  Therefore, in the near-CMC case, the near-CMC solution theory forces us into the case that
negative Yamabe solutions exist.  As we have mentioned, the uniqueness and properties of the solution
curve $((\phi(s),\bw(s)),\la(s))$ depend in large part on the first non-zero coefficient in the Taylor expansion of $\la(s)$, which depends on the value of the operator $F((\phi,\bw),\la)$
and its derivatives with respect to $\phi, \bw$ and $\la$ at $((1,{\bf 0}),0)$.  As $\tau$ does not depend on these parameters, in this case
we would not expect that there should be a connection between the uniqueness properties of solutions to $F((\phi, \bw),\la) = 0$ and the prescribed data $\tau$. 
This strongly suggests that the solution properties of the nonlinear problem $F((\phi,\bw),\la)=0$ in a neighborhood of $((1,{\bf 0}),0)$  should be the same in the near-CMC 
and far-from-CMC cases.  This line of reasoning suggests that negative Yamabe, far-from-CMC solutions exist for $\tau \in C^1(\cM)$.  However, this is merely speculation and
a rigorous analysis of the solution curves of $F((\phi,\bw),\la)=0$ needs to be done as $\tau$ varies from from a function satisfying the near-CMC condition to
one not satisfying the near-CMC assumption.

\appendix
\section{Some Supporting Results}
\label{app}
\addtocontents{toc}{\protect\setcounter{tocdepth}{-1}}

\subsection{Sobolev and H\"{o}lder norms on $\cM$}\label{norm}

Fix a smooth background metric $g_{ab}$ and let $v^{a_1,\cdots,a_r}_{b_1,\cdots,b_s}$ be a tensor of type $r+s$. 
Then at a given point $x\in \cM$, we define its magnitude to be
\begin{align}\label{eq1:8july12}
|v| = (v^{a_1,\cdots,b_s}v_{a_1,\cdots,b_s})^{\frac12},
\end{align}
where the indices of $v$ are raised and lowered with respect to $g_{ab}$.  
We then define the Banach space of $k$-differentiable functions  
$C^{k}(\cM\times \mathbb{R})$ with norm $\|\cdot\|_k$
to be those functions $u$ satisfying
$$
\|u\|_k = \sum_{j=0}^k \sup_{x\in\cM}|D^ju| < \infty,
$$
where $D$ is the covariant derivative associated with $g_{ab}$.  
Similarly, we define the space $C^k(\cT^r_s \cM)$ of $k$-times
differentiable $(r,s)$ tensor fields to be those tensors $v$ satisfying $\|v\|_k < \infty$.

Given two points $x,y \in \cM$, we define $d(x,y)$ to be the geodesic distance between them.  
Let $\alpha \in (0,1)$.  Then we may define the $C^{0,\alpha}$ H\"{o}lder seminorm for a scalar-valued function $u$ to be
$$
[u]_{0,\alpha} = \sup_{x \ne y} \frac{|u(x)-u(y)|}{(d(x,y))^{\alpha}}.
$$
Using parallel transport, this definition can be extended to $(r,s)$-tensors
$v$ to obtain the $C^{k,\alpha}$ seminorm $[u]_{k,\alpha}$ \cite{Au82}.  
This leads us to the following definition of the $C^{k,\alpha}(\cM\times \mathbb{R})$ H\"{o}lder norm
$$
\|u\|_{k,\alpha} = \|u\|_k + [u]_{k,\alpha}
$$
for scalar-valued functions, and we may define the $C^{k,\alpha}(\cT^r_s \cM)$ H\"{o}lder
norm for $(r,s)$ tensors in a similar fashion. 

Finally, we also make use in the article of the Sobolev spaces $ W^{k,p}(\cM\times \mathbb{R})$ 
and
$W^{k,p}(\cT^r_s \cM)$ where we assume
$k \in \mathbb{N}$ and $p \ge 1$.  If $dV_g$ denotes
the volume form associated with $g_{ab}$, then the 
$L^p$ norm of an $(r,s)$ tensor is defined to be 
\begin{align}\label{eq3:8july12}
\|v\|_p = \left( \int_{\cM} |v|^p dV_g\right)^{\frac1p}.
\end{align}  
We can then define the Banach space $W^{k,p}(\cM\times \mathbb{R})$ (resp. $W^{k,p}(\cT^r_s\cM)$) to be those
functions (resp. $(r,s)$ tensors) $v$ satisfying
$$
\|v\|_{k,p} = \left( \sum_{j=0}^k\|D^jv\|^p_p \right)^{\frac1p} < \infty.
$$

The above norms are independent of the
background metric chosen.  Indeed, given any two metrics
$g_{ab}$ and $\hat{g}_{ab}$, one can show that the norms induced
by the two metrics are equivalent.  For example, if $D$ and $\hat{D}$ are
the derivatives induced by $g_{ab}$ and $\hat{g}_{ab}$ respectively,
then there exist constants $C_1$ and $C_2$ such that
$$
C_1 \|u\|_{k,\hat{g}} \le \|u\|_{k,g} \le C_2 \|u\|_{k,\hat{g}},
$$
where $\|\cdot\|_{k,g}$ denotes the $C^k(\cM)$ norm with respect to $g$.  This holds for
the $W^{k,p}$ and $C^{k,\alpha}$ norms as well.  We also note that the above norms
are related through the Sobolev embedding theorem.
In particular, the spaces $C^{k,\alpha}$ and $W^{l,p}$ are related in the sense that
if $n$ is the dimension of $\cM$ and $u \in W^{l,p}$ and
$$
k+\alpha < l-\frac np, 
$$
then $u \in C^{k,\alpha}$.   See \cite{Au82, Au98, EH96,RP68} for a complete
discussion of the Sobolev embedding Theorem, Banach spaces on manifolds, and the above norms, and also~\cite{HNT09} for a numbmer of related results specifically for the constraint equations.

\subsection{Banach Calculus and Taylor's Theorem }\label{sec1:10july12}

Here we give a brief overview of some basic tools from functional analysis.
The following results are
presented without proof and are taken from \cite{EZ86};
see also~\cite{StHo2011a}.
We begin with some notation.

Suppose that $X$ and $Y$ are Banach spaces and $U \subset X$ is a neighborhood of $0$.
For a given map $f:U\subset X \to Y $, we say that 
$$
f(x) = o(\|x\|), \hspace{2mm} x\to 0 \quad \text{iff} \hspace{3mm} r(x)/\|x\| \to 0 \hspace{2mm} \text{as} \hspace{2mm} x\to 0.
$$
We write $L(X,Y)$ for the class of continuous linear maps between the Banach spaces $X$ and $Y$.

\begin{definition}
Let $U \subset X$ be a neighborhood of $x$ and suppose that $X$ and $Y$ are Banach spaces.
\begin{itemize}

\item[(1)] We say that a map $f:U \to Y$ is {\bf F-differentiable} or {\bf Fr\'{e}chet differentiable} at $x$ iff
there exists a map $T\in L(X,Y)$ such that
$$
f(x+h) - f(x) = Th+o(\|h\|), \quad \text{as} \hspace{2mm} h \to 0,
$$
for all $h$ in some neighborhood of zero.
If it exists, $T$ is called the {\bf F-derivative} or {\bf Fr\'{e}chet derivative}
of $f$ and we define $f'(x) = T$.
If $f$ is Fr\'{e}chet differentiable for all $x\in U$ we say that $f$ is Fr\'{e}chet differentiable
in $U$.
Finally, we define the {\bf F-differential} at $x$ to be $df(x;h) = f'(x)h$.

\item[(2)] The map $f$ is {\bf G-differentiable} or {\bf G\^{a}teaux differentiable} at $x$ iff there exists a map
$T\in L(X,Y)$ such that
$$
f(x+tk)-f(x) = tTk +o(t), \quad \text{as} \hspace{2mm} t \to 0,
$$
for all $k$ with $\|k\|=1$ and all real numbers $t$ in some neighborhood of zero.
If it exists, $T$ is called 
the {\bf G-derivative} or {\bf G\^{a}teaux derivative} of $f$ and we define $f'(x)=T$.
If $f$ is G-differential for all $x\in U$ we say that $f$ is G\^{a}teaux differentiable in $U$.
The {\bf G-differential} at $x$ is defined to be
$d_Gf(x;h) = f'(x)h.$
\end{itemize}
\end{definition}

\begin{remark}\label{rem1:28aug12}
Clearly if an operator is F-differentiable, then it must also be \\
G-differentiable.
Moreover, if the G-derivative $f'$ exists in some neighborhood of $x$ and $f'$ is continuous at $x$, then $f'(x)$ is also the F-derivative.
This fact is quite useful for computing F-derivatives given that G-derivatives are easier to compute.
See \cite{EZ86,StHo2011a} for a complete discussion.
\end{remark}

\noindent
We view F-derivatives and G-derivatives as linear maps $f'(x) : U \to L(X,Y)$.
More generally, we may consider higher order derivatives maps of $f$.
For example, the map ${f''(x):U \to L(X,L(X,Y))}$ is a bilinear form.
We now state some basic properties of F-derivatives.
All of the following properties also hold for G-derivatives.

The Fr\'{e}chet derivative satisfies many of the usual properties that we are accustomed to by doing calculus in $\mathbb{R}^n$.
For example, we have the chain rule.

\begin{proposition}[Chain Rule]\label{chainRule}
Suppose that $X, Y$ and $Z$ are Banach spaces and assume that $f:U \subset X\to Y$ and $g:V\subset Y \to Z$
are differentiable on $U$ and $V$ resp.\ and that $f(U) \subset V$.
Then the function $H(x) = g \circ f $, i.e.\ $H(x) = g(f(x))$, is differentiable where
$$
H'(x) = g'(f(x))f'(x)
$$
where we write $g'(f(x))f'(x)$ for $g'(f(x))\circ f'(x)$.
\end{proposition}

Given an operator $f:X\times Y \to Z$, we can also consider the partial derivative
of $f$ with respect to either $x$ or $y$.
If we fix the variable $y$ and define
$g(x) = f(x,y): X \to Z$ and $g(x)$ is Fr\'{e}chet differentiable at $x$, then the
{\bf partial derivative} of $f$ with respect to $x$ at $(x,y)$ is $f_x(x,y) = g'(x)$.
We can a make a similar definition for $f_y(x,y)$.
Finally, we observe that
we can express the F-differential of $f'(x,y)$ in terms of the partials by using
the following formula:
\begin{align}\label{eq1:7july12}
f'(x,y)(h,k) = f_x(x,y)h + f_y(x,y)k.
\end{align}
We have the following relationship between the partial derivatives and the
Fr\'{e}chet \\derivative.
\begin{proposition}
Suppose that $f:X \times Y \to Z$ is F-differentiable at $(x,y)$.
Then the partial F-derivatives
$f_x$ and $f_y$ exist at $(x,y)$ and they satisfy \eqref{eq1:7july12}.
Moreover, if 
$f_x$ and $f_y$ both exist and are continuous in a neighborhood of $(x,y)$ then $f'(x,y)$ exists as an
F-derivative and \eqref{eq1:7july12} holds.
\end{proposition}

\subsubsection{Taylor's Theorem}

As we have mentioned, the $n$-th order Fr\'{e}chet derivative of a given operator $f:X\to Y$ between Banach
spaces in a $n$-multilinear operator.  For a given $x_0 \in X$, define
\begin{align}
f^{(n)}(x_0)h^n = f^{(n)}(x_0)&(\underbrace{h,\cdots,h})\\
& \quad \text{$n$ times} \nonumber
\end{align}

Using this notation, we can state the following generalization of Taylor's Theorem for operators between Banach
spaces.  See \cite{EZ86,StHo2011a} for a proof and more details.

\begin{theorem}\label{thm1:29apr13}
Let $X$ and $Y$ be Banach spaces.  Suppose that $f:U \subset X \to Y$ is defined on an open, convex neighborhood $U$ of $x_0 \in X$.
Then if $f'(x), \cdots f^{(n)}(x)$ exist for $x \in U$, then
\begin{align}
f(x_0+h) = \sum_{n=1}^N \frac{1}{n!} f^{(n)}(x_0)h^n + R_{N+1}(x_0),
\end{align}
where
\begin{align}
\|R_{N+1}(x_0)\|_Y \le \frac{1}{(N+1)!}\sup_{0<\tau<1}\|f^{(N+1)}(x_0+\tau h)h^{N+1}\|_Y.
\end{align}
\end{theorem}

\subsection{Additional Bifurcation Theory}

In this section we present without proof, some additional results from
\cite{HK04} which are relevant to our discussion.  Proposition~\ref{prop1:6july12} presents some 
useful properties of the maps $\Phi(v,\la)$, $\psi(v,\la)$ and $\gamma(v)$
defined in the \eqref{eq6:2july12}, \eqref{eq9:6july12} and \eqref{eq2:26apr13} in Section~\ref{LSchmidt}.

\begin{proposition}\label{prop1:6july12}
Let the assumptions of Theorem~\ref{thm1:29apr12} hold and let
the operators $\Phi(v,\la)$, $\psi(v,\la)$ and $\gamma(v)$ be defined as in 
\eqref{eq6:2july12}, \eqref{eq9:6july12} and \eqref{eq2:26apr13} and let $\la_0$ and ${x_0 = v_0+w_0}$
be as in the previous discussion.  Then
\begin{align}\label{eq10:6july12}
D_v\Phi(v_0,\la_0) = 0, \quad D_v\psi(v_0,\la_0) = 0, \quad \text{and} \quad D_v\gamma(v_0) = 0,
\end{align}
and each of these operators has the same order of differentiability as $F(x,\la)$.
\end{proposition}

Once we've obtained a unique solution curve $(x(s),\la(s))$ through $(x_0,\la_0)$, we
analyze $\ddot{\la}(0)$ (where $\dot{} = \frac{d}{ds}$) to determine additional information
about the solution curve.  In particular, we can determine whether
or not a {\bf saddle node bifurcation} or fold occurs at $(x_0,\la_0)$.  This type of bifurcation
occurs when the solution curve $\{x(s),\la(s)\}$ has a turning point at $(x_0,\la_0)$.  The next proposition, also taken
from \cite{HK04}, provides us with a method to determine information about $\ddot{\la}(0)$. 

\begin{proposition}\label{prop1:1july12}
Let the assumptions of Theorem~\ref{thm1:29apr12} be in effect.   Additionally assume that 
$ker(D_XF(x_0,\la_0))$ is spanned by $\hat{v}_0$.  Then 
\begin{align}\label{eq1:1july12}
&\left.\frac{d}{ds}F(x(s),\la(s))\right|_{s=0}= \\
& D_xF(x_0,\la_0)\dot{x}(0)+D_{\la}F(x_0,\la_0)\dot{\la}(0) =D_xF(x_0,\la_0)\hat{v}_0= 0 \nonumber\\
&\left.\frac{d^2}{ds^2}F(x(s),\la(s))\right|_{s=0}= \label{eq1:17july12} \\
&D^2_{xx}F(x_0,\la_0)[\hat{v}_0,\hat{v}_0]+D_xF(x_0,\la_0)\ddot{x}(0)+D_{\la}F(x_0,\la_0)\ddot{\la}(0) = 0.\nonumber
\end{align}
In particular, an application of the projection operator
$Q$ defined in \eqref{eq1:27june12} to \eqref{eq1:17july12} yields
\begin{align}\label{eq3:5oct12}
QD^2_{xx}F(x_0,\la_0)[\hat{v}_0,\hat{v}_0]+QD_{\la}F(x_0,\la_0)\ddot{\la}(0) = 0.
\end{align}
This implies that if $D_{\la}F(x_0,\la_0)\notin R(D_xF(x_0,\la_0))$ and 
$$D^2_{xx}F(x_0,\la_0)[\hat{v}_0,\hat{v}_0] \notin R(D_xF(x_0,\la_0)),$$
then $\ddot{\la}(0) \ne 0$. 
\end{proposition}

The significance of Proposition~\ref{prop1:1july12} is that it gives explicit conditions that allow us to determine whether or not $\ddot{\la}(0)$ is nonzero.
Heuristically, the fact that $\ddot{\la}(0) \ne 0$ means that $\la(s)$ has a turning point at $s=0$.  This means
that the graph of $\{x(s),\la(s)\}$ looks like a parabola and that a ``fold" or saddle node
bifurcation occurs at $s=0$ (cf. \cite{HK04}).

\subsection{Local Representation of Conformal Equations}

Here we determine the local representation of the Hamiltionian and momentum constraints in the one-parameter family \eqref{eq1:26oct12} analyzed in
this paper.  Throughout this discussion, suppose that $g_{\la} \subset \cA^{s,p}$ $(s > 2+3/p)$ is the one-parameter family of metrics, analytic in $\la$, that is defined in Theorem~\ref{thm1:31oct12}.  
Let $\Delta_{\la}$, $\cL_{\la}$ and $D_{\la}$ be the associated 
Laplace-Beltrami, conformal Killing, and covariant derivative operators.  We begin with the following proposition, which 
describes the local representation of the one-parameter family of metrics $g_{\la} = g(x,\la)$.

\begin{proposition}\label{prop1:24may13}
The components $g_{ab}(x,\la)$ of the one-parameter family $g(x,\la)$ are analytic functions in the variable $\la$.  
Moreover, the Christoffel symbols and the coordinate derivatives of the Christoffel symbols defined by this metric are 
analytic functions in the variable $\la$.
\end{proposition}  
\begin{proof}
Because the one-parameter family $g(x,\la)$ is analytic in $\la$, we have that
\begin{align}\label{eq1:24may13}
g(x,\la) = \sum_{k=1}^{\infty}\frac{1}{k!} D_{\la}^kg(x,0)&\underbrace{(\la,\cdots,\la)}.\\
&\quad \text{$k$ times}\nonumber
\end{align}
The above expression is an infinite sum of power operators as in Definition~\ref{def2:25apr13}, where for each $k$,
\begin{align}
D_{\la}^kg(x,0): &\underbrace{\mathbb{R} \times \cdots \times \mathbb{R}} \to W^{s,p}(T^0_2\cM)\\
&\quad \text{$k$ times}\nonumber
\end{align}
is a $k$-multilinear operator from $\mathbb{R}$ to $W^{s,p}(T^0_2\cM)$.  Given that $\la \in \mathbb{R}$ and each
$D_{\la}^kg(x,0)$ is a multilinear operator, the series \eqref{eq1:24may13} can be rewritten as
\begin{align}
g(x,\la) &= \sum_{k=1}^{\infty}\frac{1}{k!} D_{\la}^kg(x,0)(\la,\cdots, \la) \\
             &= \sum_{k=1}^{\infty}\frac{1}{k!} D_{\la}^kg(x,0)(1,\cdots,1) \la^k. \nonumber
\end{align}
The above series converges in the sense of Definition~\ref{def2:25apr13}, and therefore 
converges to $g(x,\la)$ in $W^{s,p}(T^0_2\cM)$.
Furthermore, the local coordinates for $g(x,\la)$ are analytic, where
\begin{align}\label{eq2:24may13}
g_{ab}(x,\la) = g(x,\la)\left(\frac{\partial}{\partial x^a},\frac{\partial}{\partial x^b}\right) =  \sum_{k=1}^{\infty}\frac{1}{k!} D_{\la}^kg(x,0)(1,\cdots,1)\left(\frac{\partial}{\partial x^a},\frac{\partial}{\partial x^b}\right) \la^k,
\end{align}
converges in $W^{s,p}$.  Finally, because $s> 2+3/p$, the series 
$$
\sum_{k=0}^{\infty}\frac{1}{k!} \| D^k_{\la}g(x,0)\|_{C^2} |\la|^k
$$
converges, where $\| D^k_{\la}g(x,0)\|_{C^2}$ is the operator norm \eqref{eq1:20may13} induced by the norm on $\mathbb{R}$ and the norm
on $C^2(T^0_2\cM)$.  This implies that
all first and second derivatives of $g_{ab}(x,\la)$ will be analytic with respect to $\la$ and will have power series representations
that are obtained by differentiating the series \eqref{eq2:24may13} inside the sum.  
\end{proof}

We can now present the following Proposition concerning the local formulation of the family Hamiltonian constraint equations given in \eqref{eq2:11mar13}.  

\begin{proposition}\label{prop1:26feb13}
On a given coordinate chart element $(U_j,\rho_j)$, the family of Hamiltonian constraint equations
\begin{align}
-\Delta_{\la}\phi + &\frac{1}{8}\la \phi + \frac{\la^4}{12}\tau^2\phi^5-\frac{1}{8}(\la^2\sigma+\mathcal{L}{\bw})_{ab}(\la^2\sigma+\mathcal{L}\bw)^{ab}\phi^{-7}-\frac{\la^2\kappa}{4}\rho\phi^{-3}=0
\end{align}
is of the form 
\begin{align}\label{eq4:26feb13}
&\quad f_1^{ab}(\la)\partial_a\partial_b\phi+f_2^a(\la)\partial_a\phi+\frac18\la\phi+\frac{\la^4}{12}\tau^2\phi^5-\frac{\phi^{-7}}{8}\left(f_3^{abcd}(\la)\partial_aw_b\partial_cw_d \right.\\
&\left.+f_4^{abc}(\la)\partial_aw_bw_c+f^{ab}_5(\la)w_aw_b+\la^2f_6^{ab}\partial_aw_b+\la^2f_7^a(\la)w_a+\la^4\sigma^2\right)-\frac{\kappa\la^2}{4}\rho\phi^{-3} = 0 \nonumber,
\end{align}
where $f^{ab}_1,...,f^{a}_7$ are functions in $C^{2,\alpha}(U_j \times (-\delta,\delta))$ that are analytic with respect to $\la\in (-\delta,\delta)$. 
\end{proposition}
\begin{proof}
To obtain the form \eqref{eq4:26feb13}, write
\begin{align}
(\cL\bw)_{ab}(\cL\bw)^{ab} = \left(\partial_aw_b+\partial_bw_a-2\Gamma^c_{~ab}w_c-\frac23g_{ab}g^{cd}\partial_cw_d + \frac23 g_{ab}g^{dc}\Gamma^e_{~cd}w_e\right)\\
\times~~~ \left(g^{ac}g^{bd}\left(\partial_cw_d+\partial_dw_c-2\Gamma^e_{~cd}w_e-\frac23g_{cd}g^{ef}\partial_ew_f + \frac23 g_{cd}g^{fe}\Gamma^h_{~ef}w_h\right)\right)\nonumber 
\end{align}
in local coordinates, expand and group terms.  Similarly, we write $\sigma_{ab}(\cL\bw)^{ab}$ in local coordinates and then
expand.  Combining these expansions we have that \\
${(\la^2\sigma+\cL\bw)_{ab}(\la^2\sigma+\cL\bw)^{ab}}$ has the form of the expression
in the parenthesis in \eqref{eq4:26feb13}.  Writing out the local representation of the Laplace-Beltrami operator then implies that Hamiltonian
constraint has the form of Eq. \eqref{eq4:26feb13}.  That the functions $f^{ab}_1,...,f^{a}_7$ are analytic then follows from Proposition~\ref{prop1:24may13} and
and the fact that $f^{ab}_1,\cdots, f^{a}_7$ are formed from sums, products, and coordinate derivatives of the metric.  
\end{proof}

We have a similar result concerning the local representation of the family of momentum constraint equations given in \eqref{eq2:11mar13}.
\begin{proposition}\label{prop2:26feb13}
On a given coordinate chart element $(U_j,\rho_j)$, the family of momentum constraint equations
\begin{align}
\mathbb{L}_{\la}\bw +\frac{2\la^2}{3}D_{\la}\tau\phi^6+\la^2\kappa\bj =0 
\end{align}
is of the form
\begin{align}\label{eq3:26feb13}
h_1^{abcd}(\la)\partial_a\partial_bw_c+h_2^{abd}(\la)\partial_aw_b+h^{ad}_3(\la)w_a+\frac23\la^2h_4^{ad}(\la)\partial_a\tau\phi^6+\la^2\kappa j^d=0,
\end{align} 
where $h^{abcd}_1,..., h_4^{ad} \in C^{2,\alpha}(U_j\times (-\delta,\delta))$ and are analytic with respect to $\la \in (-\delta,\delta)$. 
\end{proposition}
\begin{proof}
To obtain Eq. \eqref{eq3:26feb13}, write
\begin{align}
(\mathbb{L}_{\la}\bw)_a &= D^b(\cL\bw)_{ab} = g^{bc}D_c(\cL\bw)_{ab} = g^{bc}\left(\partial_c(\cL\bw)_{ab} - \Gamma^d_{~ca}(\cL\bw)_{db}-\Gamma^d_{~cb}(\cL\bw)_{ad}\right)\\
= &g^{bc}\left[\partial_c\left( \partial_aw_b+\partial_bw_a-2\Gamma^c_{~ab}w_c-\frac23g_{ab}g^{cd}\partial_cw_d + \frac23 g_{ab}g^{dc}\Gamma^e_{~cd}w_e   \right)\right. \nonumber\\
&-\Gamma^d_{~ca}\left(\partial_dw_b+\partial_bw_d-2\Gamma^e_{~db}w_e-\frac23g_{db}g^{ce}\partial_cw_e + \frac23 g_{bd}g^{cd}\Gamma^e_{~cd}w_e             \right) \nonumber\\
&-\left.\Gamma^d_{~cb}\left(\partial_aw_d+\partial_dw_a-2\Gamma^e_{~ad}w_e-\frac23g_{ad}g^{ce}\partial_cw_e + \frac23 g_{ad}g^{cd}\Gamma^e_{~cd}w_e             \right)\right] \nonumber
\end{align}
in local coordinates, expand and group terms.  This gives the first three terms in Eq. \eqref{eq3:26feb13}, where $h^{abcd}_1, h^{abd}_2$ and $h^{ad}_3$ are formed by
sums and products of the components of $g$ and its first and second derivatives.  Therefore by Proposition~\ref{prop1:24may13} these functions will be analytic in $\la$.    
Finally, by writing the covariant derivative $D_{\la}$ in local coordinates we obtain the result.  
\end{proof}

\subsection{Far-from-CMC Existence Results}

Here we present a Theorem from \cite{HNT09} (see also~\cite{HNT08} for the smooth case), which gives conditions for which solutions to the CTT formulation
exist without the near-CMC assumption.  

\begin{theorem}\label{farCMC}
Let $(\cM,h_{ab})$ be a 3-dimensional closed Riemannian manifold suppose that $p \in (1,\infty)$ and $s \in (1+\frac3p, \infty)$ are given.  
Let $h_{ab}\in W^{s,p}$ admit no conformal Killing field and be in $\cY^+(\cM)$, the positive Yamabe class.
Select $q$ and $e$ to satisfy:
\begin{itemize}
\item $\frac1q \in (0,1)\cap (0,\frac{s-1}{3}) \cap [\frac{3-p}{3p}, \frac{3+p}{3p}],$

\item $e \in (1+\frac3q, \infty)\cap [s-1,s] \cap [\frac3q+s-\frac3p-1, \frac3q+s-\frac3p].$
\end{itemize}
Assume that the conformal data $(\tau, \sigma, \rho, \bj)$ satisfies:
\begin{align*}
&~~~\bullet~~ \tau \in W^{e-1,q} ~\text{if $e \ge 2$, and $\tau \in W^{1,z}$ otherwise, with $z=$ \tiny{$ \frac{3q}{3+\max\{0,2-e\}q}$}},                        \\ 
&~~~\bullet~~ \text{$\sigma \in W^{e-1,q}$, with $\|\sigma\|_{\infty}$ sufficiently small, }                       \\
&~~~\bullet~~ \text{$\rho \in W^{s-2,p}_+\cap L^{\infty} \backslash \{0\}$, with $\|\rho\|_{\infty}$ sufficiently small}          \\
&~~~\bullet~~ \text{$\bj \in \bW^{e-2,q}$, with $\|\bj\|_{e-2,q}$ sufficiently small.}
\end{align*}
Then there exist $\phi \in W^{s,p}$ with $\phi>0$ and $\bw \in \bW^{e,q}$ solving the Einstein constrain equations.
\end{theorem}

\bibliographystyle{abbrv}
\bibliography{Caleb,Caleb2,Caleb3}


\vspace*{-0.1cm}

\end{document}